\documentclass[copyright,creativecommons]{eptcs}
\providecommand{\event}{MSFP 2020} 

\usepackage{amsmath}
\usepackage{amssymb}
\usepackage{amsfonts}
\usepackage{fancyvrb}
\usepackage{microtype}
\usepackage[all,cmtip]{xy}
\usepackage{amsthm}
\usepackage{color}
\usepackage{graphics}
\usepackage{wrapfig}

\definecolor{darkblue}{rgb}{0.0,0.0,0.5}
\definecolor{darkgreen}{rgb}{0.0,0.4,0.0}

\usepackage[]{hyperref}
\hypersetup{
    unicode=false,          
    pdftoolbar=true,        
    pdfmenubar=true,        
    pdffitwindow=false,      
    pdftitle={},    
    pdfauthor={}
    pdfsubject={},   
    pdfnewwindow=true,      
    pdfkeywords={keywords}, 
    colorlinks=true,       
    linkcolor=darkblue,          
    citecolor=darkblue,        
    filecolor=green,      
    urlcolor=blue,          
}

\CustomVerbatimEnvironment{SVerbatim}{Verbatim}{fontsize=\small,
xleftmargin=0.5cm,
xrightmargin=0.2cm,commandchars=\\\{\},numbersep=0.9em}

\newcommand{\paragrapho}[1]{\vspace{0.25em}\noindent\textbf{#1}\hspace{0.5em}}

\newcommand{\laDP}{\textit{discrete parameterised monads}}
\newcommand{\laLM}{\textit{cat-graded monads + generalised unit}}
\newcommand{\laM}{\textit{cat-graded monads}}
\newcommand{\laMo}{\textit{monads}}
\newcommand{\laP}{\textit{parametrised monads}}
\newcommand{\laG}{\textit{unordered graded monads}}

\newcommand{\latwo}{\textit{2-cat-graded monads}}
\newcommand{\laGo}{\textit{graded monads}}
\newcommand{\latwoLM}{\textit{2-cat-graded monads + generalised unit}}

\title{Unifying graded and parameterised monads}
\def\titlerunning{Unifying graded and parameterised monads}

\author{Dominic Orchard
\institute{School of Computing, University of Kent}
\and
Philip Wadler
\institute{School of Informatics, University of Edinburgh}
\and
Harley Eades III
\institute{Department of Computer Science, Augusta University}}
\def\authorrunning{D. Orchard and P. Wadler and H. Eades}

\usepackage{amsmath}
\usepackage{amssymb}
\usepackage{stmaryrd}
\usepackage[all]{xy}
\usepackage{enumitem}
\usepackage{natbib}
\usepackage{color}

\newcommand{\op}{\mathsf{op}}
\newcommand{\Gm}{\mathsf{G}}
\newcommand{\Pm}{\mathsf{P}}
\newcommand{\F}{\mathsf{F}}
\newcommand{\D}{\mathsf{D}}
\newcommand{\G}{\mathsf{G}}
\newcommand{\T}{\mathsf{T}}
\newcommand{\TS}{\mathsf{S}}
\newcommand{\M}{\mathsf{M}}
\newcommand{\N}{\mathsf{N}}
\newcommand{\ixcat}{\mathbb{I}}
\newcommand{\jxcat}{\mathbb{J}}
\newcommand{\basecat}{\mathbb{C}}
\newcommand{\basecatD}{\mathbb{D}}
\newcommand{\inj}{\textit{in}}
\newcommand{\id}{\textit{id}}
\newcommand{\Id}{\textsf{Id}}
\newcommand{\Endo}[1]{\textbf{Endo}(#1)}
\newcommand{\sing}[1]{\Endo{#1}} 
\newcommand{\objs}[1]{#1_0}
\newcommand{\morphs}[1]{#1_1}
\newcommand{\twomorphs}[1]{#1_2}

\newcommand{\cats}{\mathbf{Cat}}
\newcommand{\interp}[1]{\llbracket{#1}\rrbracket}
\newcommand{\geneta}{\bar{\eta}}

\newcommand{\mone}{a}

\newcommand{\synLet}[3]{\textbf{let} \; #1 \leftarrow #2 \; \textbf{in} \; #3}

\newcommand{\mset}{M}
\newcommand{\mmult}{\bullet}
\newcommand{\munit}{e}

\newcommand{\icirc}{}

\newcommand{\trule}[1]{(\textit{\scriptsize{#1}})}

\newcounter{results}
\theoremstyle{definition}
\newtheorem{proposition}[results]{Proposition}
\newtheorem{corollary}[results]{Corollary}
\newtheorem{definition}[results]{Definition}
\newtheorem{notation}[results]{Notation}

\newtheorem{example}[results]{Example}
\newtheorem{remark}[results]{Remark}

\newcommand{\point}{\ast}

\usepackage{breakurl}
\newif\iflong
\longtrue

\begin{document}
\maketitle

\begin{abstract}
  Monads are a useful tool for structuring
  effectful features of computation such as state, non-determinism,
  and continuations.  In the last decade,
  several generalisations of monads have been suggested which provide a more
  fine-grained model of effects by replacing the single
  type constructor of a monad with an indexed family of constructors.
  Most notably, \emph{graded monads} (indexed by a monoid) model
  effect systems and \emph{parameterised monads} (indexed by pairs
  of pre- and post-conditions) model program logics. 
%
  This paper studies the relationship between these two
  generalisations of monads via a third generalisation.
  This third generalisation, which we call \emph{category-graded
    monads}, arises by generalising a view of monads as a particular
special case of \emph{lax functors}.
  A category-graded monad provides a family of functors $\T \, f$ indexed by
  morphisms $f$ of some other category. This allows certain
  compositions of effects to be ruled out (in the style of a program
  logic) as well as an abstract description of effects (in the style
  of an effect system).  Using this as a basis, we show how graded
  and parameterised monads can be unified,
  studying their similarities and differences along the way.
%
\end{abstract}

\section{Introduction}
\noindent

\noindent
Ever since \citet{moggi1989monads,moggi1991monads}, monads have become
an important structure in programming and semantics, particularly for
capturing computational effects. In this approach, an effectful
computation yielding a value of type $A$ is modelled by the type
$\M A$ where $\M$ is an endofunctor with the structure of a monad,
i.e., with two operations (and some axioms): \emph{unit} ($\eta$)
mapping values into pure computations (with no effects) and
\emph{multiplication} ($\mu$) composing effects of two computations
(one nested in the other):
\begin{align*}
\eta_A : A \rightarrow \M A \quad & \quad \mu_A : \M \, \M \, A \rightarrow
                                    \M \, A
\end{align*}
Various generalisations of monads replace the single endofunctor $\M$
with a family of endofunctors indexed by effect information. For
example, \citet{wadler2003marriage} proposed a generalisation of
monads to model effect systems using a family of monads
$(\M^\mone)_{\mone \in E}$ indexed over a \emph{bounded semilattice}
$(E, \sqcup, \emptyset, \sqsubseteq)$. Later, \emph{graded monads}
were proposed (\cite{katsumata2014parametric,DBLP:journals/corr/OrchardPM14}), generalising this idea
to a family of endofunctors $(\Gm \, e)_{e \in E}$ indexed by the
elements of an ordered monoid $(E, \bullet, I, \sqsubseteq)$ with monad-like
unit and multiplication operations (and axioms) mediated by the monoid
structure:
\begin{align*}
\eta_A : A \rightarrow \G \, I \, A & \qquad
\mu_{e,f,A} : \, \G \, e \,\, \G \, f \, A \rightarrow  \G \, (e \bullet f) \, A
  &
\G \, (e \sqsubseteq f) A : \G \, e \, A \rightarrow \G \, f \, A
\end{align*}
A family of approximation maps $\G \, (e \sqsubseteq f)$ is derived from the ordering.

Another indexed generalisation is provided by
\emph{parameterised monads}, which comprise a family
of endofunctors $(\Pm (I, J))_{I \in \ixcat^\op, J \in \ixcat}$ indexed by pairs of
objects drawn from a category $\ixcat$ which provides information akin
to pre- and post-conditions~(\cite{wadler1994monads,atkey2009parameterised}).
Parameterised monads have unit and multiplication operations
(satisfying analogous axioms to monads):
\begin{align*}
\hspace{-1em}
\eta_{I\!,A} \!:\! A\!\rightarrow\!\Pm (I\!, I) A & \qquad
\mu_{I\!,J,K,A} \!:\! \Pm (I\!, J) \, \Pm (J\!, K) A\!\rightarrow\!\Pm (I\!, K) A
& \Pm (f, g) A \!:\! P (I\!, J) A\!\rightarrow\! P(I'\!, J') A
\end{align*}
The family of maps $\Pm (f, g) A$
(for all $f : I' \rightarrow I, g : J \rightarrow J'$) provides
a notion of approximation via pre-condition
strengthening and post-condition weakening.
For pure computations via $\eta$, the pre-condition is
preserved as the post-condition. For composition via $\mu$,
the post-condition of the outer computation must match the
pre-condition of the inner, yielding a computation indexed by
the pre-condition of the outer and post-condition of the inner.

Graded and parameterised monads are used for various kinds of fine-grained effectful
semantics, reasoning, and programming. For example, graded monads model effect
systems
(\cite{katsumata2014parametric,mycroft2016effect}),
trace semantics (\cite{miliusgeneric}), and session types
(\cite{orchard2016effects}). Parameterised monads are used to refine effectful
semantics with Floyd-Hoare triples (\cite{atkey2009parameterised}),
information flow tracking (\cite{stefan2011flexible}), and
session types (\cite{Pucell2008,imai2011session}).
In GHC/Haskell, graded and parameterised monads
are provided by the \texttt{effect-monad}
package,\footnote{\url{https://hackage.haskell.org/package/effect-monad}}
leveraging GHC's advanced type system features particularly in the case
of graded monads.


Both structures appear to follow a similar pattern, generalising
monads into some indexed family of type constructors with monad-like
operations mediated by structure on the indices. Furthermore, there
are applications for which both graded and parameterised monads
have been used, notably {session
  types}.\footnote{\cite{orchard2017session} provide a survey of
  different structuring techniques for session types in the context of
  Haskell.} Thus, one might naturally wonder how graded and
parameterised monads are related.
We show that they can be related
by a structure we call \emph{category-graded
  monads}, based on a specialised class
of \emph{lax functors}.

Whilst graded monads are indexed by elements of a monoid and
parameterised monads by pairs of indices, category-graded monads are
indexed by the morphisms of a category, and have operations:
\begin{align*}
\!\!\!\eta_{I,A} : A \rightarrow \T \, id_I \, A \qquad &
\;\;\, \mu_{f,g,A} : \T f \, \T g \, A \rightarrow \T \, (g \circ f) \,  A
\;\;\, & \T (\textbf{f} : j \Rightarrow k) A : \T j \, A \rightarrow
         \T k \, A
\end{align*}
for $f : I \rightarrow J$, $g : J \rightarrow K$,
$j, k : I' \rightarrow J'$. Indexing by morphisms provides a way to
restrict composition of effectful computations and a model that
captures various kinds of indexing. The rightmost operation
is provided by \emph{2-category-graded monads}
(generalising grading to 2-categories)
where $\T \mathbf{f}$ lifts a
2-morphism $\mathbf{f}$ (morphism between morphisms), providing a
notion of approximation akin to graded monads.



The structures we study along with their relationships are summarised by the
diagram below, where the source of an arrow is at a more specific
structure, and the target is a more general structure.  Highlighted in
bold are the new structures in this paper, centrally category-graded
monads and its generalisation to \emph{2-category-graded monads} and the
additional structure of \emph{generalised units}.  Throughout,
we state the results of the diagram below as propositions of the form
``\emph{every A is a B}'' whenever there is an arrow from $A$ to $B$
in the diagram. By this we mean there is an injective map from $A$ structures
into $B$ structures.

\begin{wrapfigure}[10]{r}[-5px]{250px}
  \vspace{-17px}
  \footnotesize
  \begin{align*}
    \xymatrix@C=-0.6em@R=1em{
      & \textbf{\txt<7pc>{\latwoLM}} & \\
      \textbf{\txt<7pc>{\latwo}} \ar@/^0.9em/[ur]  & & \ar@/_0.9em/[ul] \textbf{\txt<7pc>{\laLM}} \\
      \laGo \ar[u]   & \ar@/_0.9em/[ul] \textbf{\laM} \ar@/^1em/[ur] & \hspace{0.5em} \laP \ar[u] \\
      *\txt<6.5pc>{\laG} \ar[u]\ar@/_1.7em/[ur] &
      & *\txt<8pc>{\laDP} \ar[u]\ar@/^1.7em/[ul] \\
      & \ar@/^/[ul] \laMo \ar[uu] \ar@/_/[ur] &
    }
  \end{align*}
\end{wrapfigure}

Section~\ref{sec:monadoids} introduces category-graded
monads as a particular way of generalising monads via lax functors.
Section~\ref{sec:graded-monads} considers graded monads and
2-category-graded monads. Section~\ref{sec:lax-nats} considers
the notion of lax natural transformations used
in Section~\ref{sec:loose} which considers parameterised
monads and generalised units.

We show that the subclasses of \emph{unordered} graded monads and
\emph{discrete} parameterised monads are both subsumed by
category-graded monads. However, graded monads may have an ordering which
provides approximation, requiring 2-category-graded monads, and full
(non-discrete) parameterised monads also have a kind of approximation
which can be captured by a generalised
unit for a category-graded monad. Section~\ref{sec:combined-example} shows that the apex of
2-category-graded monads plus generalised units capture analyses and
semantics that have both graded and parameterised monad components,
using an example of a Hoare logic for probabilistic computations.

\section{Background}
We first recall some standard, basic categorical facts relating
monoids to categories which will be used throughout.
We start by fixing some notation.

\begin{notation}[2-Category]
For a 2-category $\basecat$, we denote the class of
  objects as $\objs{\basecat}$, 1-morphisms as $\morphs{\basecat}$
  and 2-morphisms as $\twomorphs{\basecat}$. We write 2-morphism
arrows as $\Rightarrow$ unless they are natural
transformations (2-morphisms in $\mathbf{Cat}$)
which are instead written as $\xrightarrow{.}$.
Appendix~\ref{app:additional-background} provides a definition of
2-categories.
\end{notation}

\begin{notation}[Homset]
The \emph{homset} of (1-)morphisms between any two objects
$A, B \in \objs{\basecat}$ is written $\mathbb{C}(A, B)$.
\end{notation}

There are two standard ways to view monoids in categorical terms:
Recall a set-theoretic presentation of a monoid with set $\mset$, operation
$\mmult : \mset \times \mset \rightarrow \mset$ and unit $\munit \in \mset$.
This is identical to a category $\mathbb{M}$, which
has a single object $\objs{\mathbb{M}} = \{\point\}$,
 morphisms as the elements of $\mset$, i.e., $\morphs{\mathbb{M}} =
M$, composition via the binary operation $\circ = \bullet$ and
 identity morphism $\id_\point = \munit$. Associativity and unit
 properties of categories and monoids align. Hence:

\begin{proposition}\label{prop:monoid-as-singleton-cat}
 Monoids are one-object categories.
\end{proposition}

  An alternate but equivalent view of monoids is as a monoidal category, where
  $(\mset, \mmult, \munit)$ is a category with
  objects $\mset$, bifunctor $\mmult : \mset \times \mset \rightarrow
  \mset$ and unit object $\munit$. This monoidal category is \emph{discrete}, meaning
  the only morphisms are the identities.

\begin{proposition}\label{prop:monoid-as-discrete-monoid-cat}
  Monoids are discrete monoidal categories.
\end{proposition}

In this paper, we study graded monads
which are typically defined in the literature as being
indexed by a pre-ordered monoid, or \emph{pomonoid}.
We therefore recall how the above two results generalise to pomonoids.
A monoid $(\mset, \mmult, \munit)$ with a preorder
$\sqsubseteq$ (with $\mmult$ monotonic wrt.
$\sqsubseteq$)\footnote{That
is for all $x_1, x_2, y_1, y_2 \in M$ then
$x_1 \sqsubseteq y_1\ \wedge\
 x_2 \sqsubseteq y_2 \implies (x_1 \mmult x_2) \sqsubseteq (y_1 \mmult y_2)$}
is a 2-category following
Proposition~\ref{prop:monoid-as-singleton-cat}
but with added 2-morphisms between every pair of morphisms
$x, y \in \mset$ whenever $x \sqsubseteq y$.

\begin{proposition}\label{prop:pomonoid-as-2-cat}
Pre-ordered monoids are one-object 2-categories.
\end{proposition}

  Pre-orders are categories, with a morphism for every pair of
  ordered elements. Thus we can replay
  Proposition~\ref{prop:monoid-as-discrete-monoid-cat}, but we
  now have morphisms between some elements representing the ordering and thus
  the category is no longer discrete. However this monoidal
  category is \emph{strict}, meaning that associativity and unit axioms
  are equalities rather than isomorphisms. (Note, {discrete}
  categories are automatically strict).

\begin{proposition}\label{prop:pomonoid-as-strict-monoidal-categories}
  Pre-ordered monoids are strict monoidal categories.
\end{proposition}

A further categorical view on monoids is that they can be
identified as some distinguished object in a monoidal category:
\begin{definition}\label{def:monoid-in-monoid-cat}
  A monoid in a monoidal category $(\mathbb{C}, \otimes, I)$
  is a distinguished object $M \in \mathbb{C}$ equipped
  with a pair of morphisms $\munit : I \rightarrow M$
  for the unit and $\mmult : M \otimes M \rightarrow M$,
  with associativity and unit axioms.
\end{definition}
The \emph{monoidal category of endofunctors} for some category
is particularly important in this paper:

\begin{definition}\label{prop:endofunctors-are-monoids}
  The category of endofunctors on $\basecat$,
  denoted $[\basecat, \basecat]$, has endofunctors
  as objects and natural transformations as morphisms.
  This is a strict monoidal category with
  $([\basecat, \basecat], \circ,
  \mathsf{Id}_\basecat)$, i.e., with functor composition
  as the bifunctor and the identity endofunctor
  $\mathsf{Id}_\mathbb{C} A = A$ as the unit element.
\end{definition}

The classic aphorism that \emph{monads are just a monoid in the
  category of endofunctors} thus applies
Definition~\ref{def:monoid-in-monoid-cat}
in the context of the monoidal
category of endofunctors $([\basecat, \basecat], \circ,
\mathsf{Id}_\basecat)$, pointing out that a monad for
endofunctor $\T : \basecat \rightarrow \basecat$ is a
particular single object. Since
the tensor product of the monoidal category is $\circ$ then monad
multiplication is the binary operator $\mu : \T \circ \T
\xrightarrow{.} T$ and the monad unit operation identifies the unit
element $\eta : \Id_{\basecat} \xrightarrow{.} T$. Thus:

\begin{proposition}\label{aphorism}
Monads are monoids in the category of endofunctors.
\end{proposition}

Finally, we leverage the (standard)
equivalence between strict monoidal categories and
2-categories with one object, which
``transposes'' monoidal and 1-categorical composition into
1-categorical and 2-categorical composition respectively.
That is, for a strict monoidal category $(\mathbb{C}, \otimes, I)$
  we can construct a one-object 2-category, call it $1(\mathbb{C})$, with
  $\objs{1(\mathbb{C})} = \point$ and
  $\morphs{1(\mathbb{C})} = \objs{\mathbb{C}}$
    and $\twomorphs{1(\mathbb{C})} = \morphs{\mathbb{C}}$
    with horizontal composition $\circ_0 = \otimes$
    and identity morphism $id = I$,
    and vertical composition $\circ_1 = \circ_{\mathbb{C}}$
    and 2-identities by identities of $\mathbb{C}$.
    Conversely, given a one-object 2-category, $\mathbb{C}$, we can
    construct a strict monoidal
    category $(\mathsf{SMC}(\mathbb{C}), \otimes, I)$ where
    $\mathsf{SMC}(\mathbb{C})_0 = \mathbb{C}_1$, $\mathsf{SMC}(\mathbb{C})_1 = \mathbb{C}_2$,
    $\otimes = \circ_0$, $I$ is the 1-morphism identity,
    composition $\circ = \circ_1$ and identity is the 2-identity.
    The same result applies for discrete monoidal categories
    and one-object categories
    where we can elide the 2-categorical part. Hence:

 \begin{proposition}\label{cor:monoidal-cats-as-single-cats}
  Discrete monoidal categories are equivalent
  to one-object categories, and
  strict monoidal categories are equivalent to one-object
  2-categories.
\end{proposition}

\begin{corollary}\label{cor:endo-c}
  By Def.~\ref{prop:endofunctors-are-monoids} and
  Prop.~\ref{cor:monoidal-cats-as-single-cats},
  the strict monoidal category of endofunctors $([\basecat, \basecat], \circ,
  \mathsf{Id}_{\basecat})$ is equivalent to a one-object 2-category
  with single object $\mathbb{C}$,
  morphisms as endofunctors and 2-morphisms are natural
  transformations. For clarity, we denote this
  2-category as $\sing{\mathbb{C}}$.
\end{corollary}


\section{Generalising monads via lax functors to category-graded monads}
\label{sec:monadoids}
Recall that every monoid corresponds to a single-object category
 (Proposition~\ref{prop:monoid-as-singleton-cat}). 
In a single-object category, all morphisms
compose just as all elements of a monoid multiply.  Categories
can therefore be seen as a generalisation of monoids: morphisms $f$ and $g$ compose to
$g \circ f$ only when the source of $g$ agrees with the target of $f$.
Similarly, \emph{groupoids} generalise the notion of groups to
categories, where elements of the group are morphisms and every
morphism has an inverse.

This process of generalising some notion to a category is known
as \emph{horizontal categorification} 
or \emph{oidification} (echoing the relationship of groups to
groupoids) (\cite{nlab:horizontal_categorification,bertozzini2008horizontal,bertozzini2008non}). The general approach is to realise some concept as
a kind of category comprising a single object which can then be
generalised to many objects; categories are the ``oidification'' of
monoids going from a single object to many
(\cite{nlab:horizontal_categorification})---though the term \emph{category} is preferred to
\emph{monoidoid}!

Our approach can be summarised as the horizontal categorification
of monads. It turns out that this yields structures that unify
graded and parameterised monads, but with some subtleties as we shall see.

First we need to understand in what way the concept of a monad can be
seen as single-object entity such that it can be subjected to
oidification. The view of a monad as a monoid in
category of endofunctors (Proposition~\ref{aphorism}) highlights
that a monad comprises a single object in $[\basecat, \basecat]$ but
oidification cannot readily be applied to this perspective.  Instead,
we take B\'{e}nabou's \citeyearpar{benabou1967introduction} view of monads as \emph{lax functors} which is
more amenable to oidification. We first recall the definitions of
lax functors for completeness.

\begin{definition}
\label{def:lax-functor}
  A \emph{lax} functor $\F : \basecat{} \rightarrow \basecatD{}$
  (where $\basecatD{}$ is a 2-category)\footnote{\emph{Lax functors}
  are often defined between 2-categories, but can instead be defined (like here) with
  a 1-category in the domain which is treated as a 2-category with
  only trivial 2-morphisms.}
  comprises
  an object mapping and a morphism mapping, however the usual
  functor axioms $id_{\F A} = \F id_A$ and $\F g \circ \F f = \F
  (g \circ f)$ are replaced by families of 2-morphisms in
  $\basecatD{}$ which ``laxly'' preserve units
  and composition:
\begin{equation*}
 \eta_A : \; id_{\F A} \Rightarrow \F id_A
 \qquad
 \mu_{f,g} : \; \F g \circ \F f \Rightarrow \F (g \circ f)
\end{equation*}
We have chosen the names of these families to be suggestive of our endpoint here.

Whilst the axioms of a category are preserved automatically by
non-lax functors, this is not the case here.
For example, if $\F$ is a functor then
$\F f = \F (id \circ f) = \F id \circ \F f = \F f$, but
not if $\F$ is a lax functor. Instead a lax functor has
additional axioms for associativity of $\mu$ and unitality of
$\eta$:
\begin{equation*}
    \xymatrix@C=2.4em@R=1.75em{
      \F f  \ar[r]^-{\text{\footnotesize{$\eta_{J}\,\F f$}}} \ar[d]_{\text{\footnotesize{$\F f \, \eta_{I}$}}} \ar@{=}[dr]
      &  \F {id_J} \circ \F f \ar[d]^{\text{\footnotesize{$\mu_{\id_J,f}$}}} \\
      \F f \circ \F {\id_I} \ar[r]_{\text{\footnotesize{$\mu_{f,\id_I}$}}} & \F f
    }
    \quad\qquad
    \xymatrix@C=3.2em@R=1.75em{
      \F f \circ \F g \circ \F h  \ar[r]^{\text{\footnotesize{$\F f \, \mu_{g,h}$}}} \ar[d]_{\text{\footnotesize{$\mu_{f,g} \, \F h$}}}
      & \F f \circ \F {(g \circ h)} \ar[d]^{\text{\footnotesize{$\mu_{f,g \circ h}$}}} \\
      \F {(f \circ g)} \circ \F h \ar[r]_{\text{\footnotesize{$\mu_{f \circ g,h}$}}} & \F (f \circ g \circ h)
    }
  \end{equation*}
\end{definition}
\noindent
We can thus see monads as lax functors between the terminal category
$\mathbf{1}$ and the one-object 2-category of endofunctors on
$\mathbb{C}$:
\begin{proposition}[\cite{benabou1967introduction}]
  \label{def:sing-cat}
  For a category $\basecat{}$, a \emph{monad} on $\basecat{}$
  is a lax functor
  $\T : \mathbf{1} \rightarrow \sing{\basecat}$
 where $\mathbf{1}$ is the single-object category with
  $\point \in \objs{\mathbf{1}}$ and a single morphism
  $id_\point : \point \rightarrow \point$. Then $\T \point = \mathbb{C}$
  and $\T \id_\point$ identifies an endofunctor on $\basecat$.
  Laxness means the functor axioms for $\T$ are 2-morphisms, which
  are the unit and multiplication operations of the monad on
  the endofunctor $\T \id_\point$:
\begin{equation*}
\eta : id_{\T \point} \xrightarrow{.} \T \, {\id_\point}
\qquad
\mu : \T \, {\id_\point} \icirc \T \, {\id_\point} \, \xrightarrow{.} \, \T \, {\id_\point}
\end{equation*}
where $id_{\T \point}$ is the identity endofunctor $\Id$ on $\basecat{}$.
The monad axioms are exactly the unit and associativity axioms
of the lax functor.
\end{proposition}
This proposition recasts the aphorism that \emph{monads are monoids in the
  category of endofunctors}. It equivalently views monads as
lax homomorphisms (lax functors)
between the singleton monoid $\mathbf{1}$ and the
monoid of endofunctors.
Since the source category
$\mathbf{1}$ has but one object and one morphism, $\T$
identifies a particular endofunctor on $\mathbb{C}$
and the lax operators $\eta$ and $\mu$ provide the usual
monad operations for $\T$.

We can now ``oidify'' this
definition in two ways: (1) generalising the singleton source
  $\mathbf{1}$ to a category, or (2) generalising the singleton target
  category $\sing{\basecat}$ to $\cats$.
In this paper, we pursue
the first choice, though we discuss the second in
Section~\ref{subsec:oidification-alt}.
We thus oidify B\'{e}nabou's monad definition by replacing the
singleton category $\mathbf{1}$ with an arbitrary category $\ixcat$.  We might
have chosen the name \emph{monadoid} for this structure, but since
there are two ways in which the oidification can be applied, we settle
on the name \emph{category-graded monads}, the terminology for which
will be explained once we recall {graded monads} in the next
section. We also found that people do not like the name
\emph{monadoid}.

\begin{definition}
  Let $\ixcat$ and $\basecat$ be categories.
  A \emph{category-graded monad}
  is a {lax functor}
  $\T : \ixcat{}^\op \rightarrow \Endo{\basecat}$, and because $\Endo{\basecat}$ has only a single object we require that $\T \, I = \mathbb{C}$ for all $I \in \objs{\mathbb{C}}$.
  Thus the morphism mapping of $\T$ can be thought
of as family of endofunctors indexed by $\ixcat{}$ morphisms, i.e.,
for all $f : I \rightarrow J$ in $\ixcat$ then
  $\T \, f : \basecat \rightarrow \basecat$ is an endofunctor.

We refer to $\ixcat$ as the indexing category
and often write $\T f$ as $\T_f$. For brevity we sometimes refer to
category-graded monads
as \emph{cat-graded monads} or
 \emph{$\ixcat$-graded monads} if
the category $\ixcat$ is in scope.

\label{def:monadoid}
\end{definition}
The definition of category-graded monads is compact, but the requirement
that $T$ is a lax functor comes with a lot of structure
which is akin to that of monads and graded monads.
\begin{corollary}[Category-graded monad operations and axioms]
  \label{corollary:category-graded_monad_laws}
  Suppose $\T : \ixcat{}^\op \rightarrow \Endo{\basecat}$ is a
  category-graded monad as defined above.  Then, following from the lax functor
  definition for $\T$ there are natural transformations
  (which we call \emph{unit} and \emph{multiplication} respectively):
  \begin{equation*}
    \eta_I : id_{\T I} \xrightarrow{.} \T_{id_I}
    \qquad
    \mu_{f, g} : \T_f \icirc \T_g \xrightarrow{.} \T_{(g \circ f)}
    \qquad \text{(\textit{where} $f : I \rightarrow J$ and $g : J \rightarrow K$ are
      in $\ixcat$)}
  \end{equation*}
  Following from the lax functor,
  these satisfy associativity and unitality axioms
  (specialised from Def.~\ref{def:lax-functor}):
  \begin{equation*}
    \xymatrix@C=2.4em@R=1.5em{
      \T_f  \ar[r]^-{\text{\footnotesize{$\eta_{J} \T_f$}}} \ar[d]_-{\text{\footnotesize{$\T_f \, \eta_{I}$}}} \ar@{=}[dr]
      &  \T_{id_J} \icirc \T_f \ar[d]^-{\text{\footnotesize{$\mu_{\id_J,f}$}}} \\
      \T_f \icirc \T_{\id_I} \ar[r]_{\text{\footnotesize{$\mu_{f,\id_I}$}}} & \T_f
    }
    \quad\qquad
    \xymatrix@C=3.2em@R=1.5em{
      \T_f \icirc \T_g \icirc \T_h  \ar[r]^{\text{\footnotesize{$\T_f \mu_{g,h}$}}} \ar[d]_{\text{\footnotesize{$\mu_{f,g}\T_h$}}}
      & \T_f \icirc \T_{(h \circ g)} \ar[d]^{\text{\footnotesize{$\mu_{f,h \circ g}$}}} \\
      \T_{(g \circ f)} \icirc \T_h \ar[r]_{\text{\footnotesize{$\mu_{g \circ f,h}$}}} & \T_{h \circ g \circ f}
    }
  \end{equation*}
  Note that the left square uses the equality $\T_{(f \circ \id_I)} =
  \T_{(\id_J \circ f)} = \T_f$ due to the unitality of $\id$ and the
  right square uses associativity of composition such that $\T_{h
    \circ (g \circ f)} = \T_{(h \circ g) \circ f}$
\end{corollary}

\begin{example}
  Following Prop.~\ref{def:sing-cat} and the above definition,
  every monad ($\M : \mathbb{C} \rightarrow \mathbb{C}, \mu, \eta)$
is a category-graded monad on $\mathbb{C}$, with
indexing $\ixcat = \mathbf{1}$,
lax functor $\T_{id_\point} = \M$, and
operations $\mu_{id_{\point},id_{\point}} = \mu$
 and $\eta_{\point} = \eta$.
\label{prop:simple-monad-to-smonad}
\end{example}

\begin{remark}
For morphisms $f : I \rightarrow J, g : J \rightarrow
K$ in $\ixcat$, any two morphisms $j : A \rightarrow \T_f\, B$ and $k : B \rightarrow
\T_g\, C$ in $\basecat$ can then be composed using multiplication:
\begin{align}
      {A\ \xrightarrow{j}\ \T_f\, B\ \xrightarrow{\T_f\, k}\ \T_f \T_g\, C\
        \xrightarrow{\mu_{f,g,C}}\ \T_{g \circ f}\, C}
\label{eq:monadoid-comp}
\end{align}
\noindent
This composition is akin to Kleisli composition of a monad and has
 identities given by $\eta_{I,A}$. This
 shows the role of the opposite category in $\T : \ixcat{}^\op
\rightarrow \sing{\basecat}$: the outer functor in the source
of $\mu$ ($\T_f$ of $\T_f \T_g$) corresponds to the first effect and the inner
($\T_g$ of $\T_f \T_g$) corresponds to the second effect.
Sequential composition via $\mu_{f,g,C}$ then returns an object
 $\T_{g \circ f}\, C$.
\label{rem:monadoid-comp}
\end{remark}

\iflong
\begin{remark}
  The above composition for morphisms of the form
  $A \rightarrow \T_f\, B$ corresponds
  to the \emph{Grothendieck construction} for an indexed category
  ${\T '} : \ixcat^\op \rightarrow
  \mathbf{Cat}$~(\cite{grothendieck1961categories}) which maps
  an $\ixcat^\op$-indexed category $\T '$ to a category $\ixcat{}\!\int{}\!\T '$
  which provides a category with:
  \begin{itemize}
  \item objects
    $\objs{(\ixcat{}\!\int{}\!\T ')} = \objs{\ixcat} \times
    \objs{\mathbb{C}}$ i.e., pairs of index and base category objects;

   \item morphisms $g : I \times A \rightarrow J \times B$ given
     by $g = f \times h$ where
     $f : I \rightarrow J \in \morphs{\ixcat}$ and
     $h : A \rightarrow \T '_f \, B \in \morphs{\basecat}$;

     \item composition and identities defined as in
       Remark~\ref{rem:monadoid-comp} in terms of the lax functor
       operations $\mu$ and $\eta$.
  \end{itemize}
As pointed out by~\citet{jacobs1999categorical},
$\mu$ and $\eta$ need not be natural isomorphisms
but just natural transformations in order for the above construction
to be a category~\cite[p.117, 1.10.7]{jacobs1999categorical}; that is,
$\T '$ need only be a lax functor for the Grothendieck construction to
work, as is the case for category-graded monads here. 
\label{rem:grothendieck}
\end{remark}
\fi

\begin{example}
The \emph{identity category-graded monad} is a lifting of the identity monad
to an arbitrary indexing category. We denote this $\textbf{Id} :
\ixcat^\op \rightarrow \Endo{\basecat}$ with $\textbf{Id}_f =
\Id_\basecat$. This will be useful later.
\label{exm:identity-monadoid}
\end{example}

\begin{example}
  \label{exm:static-trace}
A category can be viewed as a state machine: objects
as states and morphisms as transitions. Given a monad $\M$
on $\basecat$
and some category $\ixcat$,
there is a category-graded monad $\T : \ixcat^\op \rightarrow \Endo{\basecat}$
with $\T_f = \M$, restricting the
composition of effectful computations by the
morphisms of $\ixcat$. Any effectful operation producing values $\M A$ can be
suitably wrapped into $\T_f \, A$ for some particular $f$ describing the operation
and its corresponding state transition. A cat-graded monad model for a program then
provides a \emph{static trace} of the effects in its indices.
We may also wrap a monad into a category-graded monad
but with some additional implementation
related to the grading, with the implementation dependent on the indices.

For example, we could capture the simple state-machine
protocol of a mutual exclusion
lock over a memory cell with $\objs{\ixcat} = \{\textsf{free},
\textsf{critical}\}$
and morphisms $\textsf{lock} : \textsf{free} \rightarrow
\textsf{critical}$ and $\textsf{unlock} : \textsf{critical}
\rightarrow \textsf{free}$ and $\textsf{get}, \textsf{put} :
\textsf{critical} \rightarrow \textsf{critical}$. We can then
wrap the state monad, call it $\mathsf{St}^S$
for some state type $S$, to get an $\ixcat$-graded monad
$\mathsf{ConcSt}^S$, wrapping the usual state monad operations
as $\emph{get} : \mathsf{ConcSt}^S_{\textsf{get}}\ S$ and
  $\emph{put} : S \rightarrow \mathsf{ConcSt}^S_{\textsf{put}}\ 1$
  and adding operations (implemented in $\mathsf{ConcSt}$)
  $\emph{lock} : \mathsf{ConcSt}^S_{\textsf{lock}}\ 1$ and
  $\emph{unlock} : \mathsf{ConcSt}^S_{\textsf{unlock}}\ 1$
  and an operation for spawning threads
  from computations whose index is any morphism from
  $\textsf{free}$ to $\textsf{free}$, i.e.
  $\emph{spawn} : (\forall f . \mathsf{ConcSt}^S_{f : \mathsf{free}
    \rightarrow \mathsf{free}} 1) \rightarrow \mathsf{ConcSt}^S_{id_{\mathsf{free}}}
  1$. Subsequently,
  the category grading statically ensures the mutual exclusion protocol:
  a spawned thread must acquire the \emph{lock} before it can \emph{get} or
  \emph{put} (or indeed \emph{unlock}), and the lock must be
  released if acquired.
  In Haskell-style syntax, we can then have programs like:
  $\textbf{do} \{\textit{lock};\ x \leftarrow \textit{get};\
  \textit{put} (x + 1);\ \textit{unlock}\} :
  \mathsf{ConcSt}^{\textsf{Int}}_{\mathsf{unlock} \circ \mathsf{put}
    \circ \mathsf{get} \circ \mathsf{lock} : \mathsf{free} \rightarrow \mathsf{free}}\ 1$.
\end{example}

\section{Graded monads as category-graded monads}
\label{sec:graded-monads}
Graded monads are a generalisation of monads from a single endofunctor
to an indexed family of endofunctors whose indices are drawn from a
(possibly ordered) monoid. The graded monad operations are
``mediated'' by the monoidal structure of the indices.
Graded monads appear in the literature under various names:
\emph{indexed monads}~(\cite{DBLP:journals/corr/OrchardPM14}), \emph{parametric
  effect monads}~(\cite{katsumata2014parametric}) or \emph{parametric
  monads}~(\cite{mellies2012parametric}).  The terminology of
\emph{graded monads} (e.g., in
\cite{smirnov2008graded,fujii2016towards,miliusgeneric,gaboardi2016})
is now the preferred term in the community.

Here we show
that lax functors $\T : \ixcat{}^\op \rightarrow \sing{\basecat}$
generalise the notion of graded monads. We thus explain
why we chose the name \emph{category-graded monads}, since
this structure can be seen as generalising a monoid-based
grading to a category-based grading.
The idea of generalising
graded monads to lax functors is mentioned by
\cite[\S6.2]{katsumata2014parametric} and \cite{fujii2016towards}. We give
the full details.

\begin{definition}
  Let $(\mset, \mmult, \munit, \sqsubseteq)$ be a
  pomonoid, meaning that $\mmult$ is also monotonic
  with respect to $\sqsubseteq$, presented as a strict monoidal
category (Prop.~\ref{prop:pomonoid-as-strict-monoidal-categories})
  which we denote as $\mset$ for convenience.

 A \emph{graded monad} comprises a functor
  $\Gm : \mset \rightarrow [\basecat,\basecat]$ which is \emph{lax
    monoidal} and therefore has
  natural transformations witnessing lax preservation of the
  monoidal structure of $\mset$ into the monoidal structure of
  $[\basecat,\basecat]$, that is
\emph{unit} $\eta^{\Gm} : \Id \xrightarrow{.} \Gm \, \munit$
and \emph{multiplication}
$\mu^{\Gm}_{m,n} : \Gm \, m \, \icirc \Gm \, n \, \xrightarrow{.}
\Gm \, (m \mmult n)$. The morphism mapping of the functor $\Gm$ means
that an ordering $m \sqsubseteq n$ corresponding to a
morphism $f : m \rightarrow n \in \morphs{\mset}$ is then mapped to a
natural transformation $\Gm f : \Gm \, m \xrightarrow{.} \Gm \, n$\
which we call an \emph{effect approximation}.
Lax monoidality of $\Gm$ means that functoriality of
$\bullet$ is laxly preserved by $\mu^{\Gm}$. A graded monad thus
satisfies axioms:
\begin{align*}
\xymatrix@C=2em@R=1.75em{
 \Gm \, m \ar[r]^-{\text{\footnotesize{$\Gm \eta$}}} \ar[d]_{\text{\footnotesize{$\eta \Gm$}}} \ar@{=}[dr] &
 \Gm \, m \icirc \Gm \, \munit \ar[d]^{\text{\footnotesize{$\mu_{m, \munit}$}}} \\
 \Gm \, \munit \icirc \Gm \, m \ar[r]_{\text{\footnotesize{$\mu_{\munit, m}$}}} &  \Gm \, m
}
\quad
\xymatrix@R=1.75em@C=2.75em{
 \Gm \, m \icirc \Gm \, n \icirc \Gm \, p \ar[r]^-{\text{\footnotesize{$\Gm \mu_{n,p}$}}} \ar[d]_{\text{\footnotesize{$\mu_{m,n} \Gm$}}}
&  \Gm \, m \icirc \Gm \, (n \mmult p) \ar[d]^{\text{\footnotesize{$\mu_{m, n \mmult p}$}}} \\
\Gm \, (m \mmult n) \icirc \Gm \, p \ar[r]_{\text{\footnotesize{$\mu_{m \mmult n, p}$}}} &
\Gm \, (m \mmult n \mmult p)
}
\quad
\xymatrix@R=1.75em@C=2em{
\Gm n \icirc \Gm m  \ar[d]_-{\text{\footnotesize{$\mu^{\Gm}_{n,m}$}}}
\ar[r]^-{\text{\footnotesize{$\Gm f \Gm$}}} &
\Gm n' \icirc \Gm m
\ar[r]^-{\text{\footnotesize{$\Gm \Gm g$}}} &
\Gm n' \icirc \Gm m' \ar[d]^-{\text{\footnotesize{$\mu^{\Gm}_{n',m'}$}}}  \\
\Gm (n \bullet m)
\ar[rr]_{\text{\footnotesize{$\Gm (f \bullet g)$}}} & & \Gm (n' \bullet m')
}
\end{align*}
where $f : n \rightarrow n', g : m \rightarrow m' \in \mset_1$
in the rightmost diagram.

Note, graded monads are not families of monads;
$\Gm \, m$ need not be a monad for all $m$.
For example:
\label{def:graded-monad}
\end{definition}

\begin{example}
The \emph{graded list monad} is indexed by the
 monoid $(\mathbb{N}, \ast, 1, \leq)$ refining the usual list monad by
length $\Gm_ n \, A = A^0 + A^1 + \ldots + A^n$ thus
$\Gm_n A$ represents the type of lists of length
at most $n$ with elements of type $A$. This graded monad then has the operations
$\eta^{\Gm}_A : A \rightarrow \Gm_1 A = \mathsf{in}_1$
injecting a value into a singleton list and
$\mu^{\Gm}_{n, m,A} : \Gm_m (\Gm_n A)
\rightarrow \Gm_{m \ast n} A = \textit{concat}$
which concatenates together a list of lists
and
approximation $\Gm (f : m \leq n) : \Gm_m A \rightarrow \Gm_n A$
which views a list of at most length $m$ as a list
of at most length $n$ when $m \leq n$.
\end{example}

In the literature, graded monads need not have a pre-ordering
(e.g.,~\cite{mycroft2016effect,gibbons2016comprehending}), so we may
distinguish two kinds of graded monad depending on the ``grading''
structure: \emph{monoid-graded monads} and \emph{pomonoid-graded
  monads}. A monoid-graded monad
$\Gm : A \rightarrow [\basecat, \basecat]$ therefore maps from a
\emph{discrete} strict monoidal category, i.e., one whose only morphisms are the
identities. Generally when we refer to ``graded monads'' we
mean pomonoid-graded monads since this is the most common meaning in
the literature. We qualify the
nature of the grading structure when we need to be explicit.

We now come to the first main result: that monoid-graded monads are a
simple case of category-graded monads via the old idea of monoids as
single-object categories (Prop.~\ref{prop:monoid-as-singleton-cat})
and monoidal categories (Prop.~\ref{prop:monoid-as-discrete-monoid-cat}):

\begin{proposition}
  Every monoid-graded monad is a category-graded monad.
  \label{prop:gmonad-to-monadoid}
\end{proposition}
\begin{proof}
  The indexing category for a monoid-graded monad on $\mathbb{C}$ is discrete
  monoidal $(\mset, \mmult, \munit)$, thus it is equivalent to a
  single-object category (Prop.~\ref{cor:monoidal-cats-as-single-cats}),
  which we write as ${1}(\mset)$.
  Similarly, $[\basecat, \basecat]$ is
  equivalent to the single object 2-category $\sing{\basecat}$
  (Cor.~\ref{cor:endo-c}).  Thus
  $\Gm : \mset \rightarrow [\basecat, \basecat]$ is equivalent to a
  lax functor $\T : {1}(\mset) \rightarrow \sing{\basecat}$,
  with $\T f = \Gm f$ (where $f$ is a morphism of ${1}(\mset)$
  and an object of $\mset$) and $\eta_{1} \, = \eta^{\Gm}$ and
  $\mu_{f,g} \, = \mu^{\Gm}_{f,g}$ whose lax functor axioms follow
  directly from the graded monad axioms.
\end{proof}

\paragraph{On the naming}
The general paradigm of ``grading'' is to have an indexed structure
where the structure of the indices matches the structure of some
underlying semantics or term calculus (\cite{10.1145/3341714}).
We can grade by different structures, mapping (e.g.,
as a homomorphism or functor) to the underlying domain (in our case
$\sing{\basecat}$). Thus, we can view this particular class of
lax functors as a horizontal categorification of monads
which serves to ``grade'' by a category, rather than grading by a single object
as with monads or grading by a (po)monoid as with traditional graded monads.

\subsection{Pomonoid-graded monads and 2-category-graded monads}
\label{sec:pomonoid-graded}

Proposition~\ref{prop:gmonad-to-monadoid} considers only monoid-graded
monads (without an ordering) however in the literature
\emph{pomonoid}-graded monads are the norm. Since strict monoidal
categories are equivalent to one-object 2-categories
(Prop.~\ref{cor:monoidal-cats-as-single-cats}),
where the
morphisms of the former become the 2-morphisms of the latter, we
therefore generalise category-graded monads to 2-category-graded
monads to complete the picture.

\begin{definition}
\label{def:2-monadoid}
A \emph{2-category-graded monad} extends a category-graded monad to a
2-categorical index category, thus
$\T : \ixcat{}^\op \rightarrow \sing{\basecat}$ is a lax
2-functor. This provides a 2-morphism mapping from $\twomorphs{\ixcat}$ to natural
transformations on $\basecat$: for all
$f, g : A \rightarrow B \in \morphs{\ixcat}$ and
$\mathbf{f} : f \Rightarrow g \in \twomorphs{\ixcat}$ then
$\T_{\mathbf{f}} : \T_f \xrightarrow{.} \T_g$ is a natural
transformation which we call an \emph{approximation} to recall
the similar idea in graded monads.

The 2-morphism mapping has 2-functorial
  axioms in two flavours: for vertical composition (left two
diagrams) and for horizontal composition which is laxly
preserved by $\T$ (right two diagrams):
\newcommand{\kK}{\mathbf{k}}
\newcommand{\gG}{\mathbf{g}}
\newcommand{\fF}{\mathbf{f}}
\newcommand{\hH}{\mathbf{h}}
\begin{align*}
\xymatrix@C=1.75em@R=1.25em{
\T_f \ar[d]_{\text{\footnotesize{$\T_{\mathbf{id}_f}$}}}
\ar[dr]^{\text{\footnotesize{$\mathbf{id}_{\T f}$}}} \\
\T_f \ar@{=}[r] & \T_f
}
\qquad
\xymatrix@C=1.75em@R=1.25em{
\T_f \ar[d]_{\text{\footnotesize{$\T_{\mathbf{f}}$}}} \ar[dr]^{\text{\footnotesize{$\T_{\mathbf{g}\, \circ_1\, \mathbf{f}}$}}} &  \\
\T_g \ar[r]_{\text{\footnotesize{$\T_{\mathbf{g}}$}}} & \T_h
}
\qquad
\xymatrix@C=1.75em@R=1.25em{
\Id_{\mathbb{C}} \ar[dr]^{\text{\footnotesize{$\eta_I$}}} \ar[d]_{\text{\footnotesize{$\eta_I$}}} \\
\T_{id_I} \ar[r]_{\text{\footnotesize{$\T_{\mathbf{id}_{id_I}}$}}} & \T_{id_I}
}
\qquad
\xymatrix@C=1.75em@R=1.25em{
\T_f \T_g \ar[d]_{\text{\footnotesize{$\mu_{f,g}$}}}
\ar[r]^-{\text{\footnotesize{$\T_{\mathbf{f}} \T_{\mathbf{g}}$}}} &
\T_{f'} \T_{g'} \ar[d]^{\text{\footnotesize{$\mu_{f',g'}$}}} \\
\T_{g \circ f} \ar[r]_-{\text{\footnotesize{$\T_{\mathbf{g}\, \circ_0\, \mathbf{f}}$}}}
& \T_{g' \circ f'}
}
\end{align*}
The left two diagrams hold
for all $f, g, h : A \rightarrow B \in \ixcat_1$
and $\mathbf{f} : f \Rightarrow g \in \ixcat_2$ and
$\mathbf{g} : g \Rightarrow h \in \ixcat_2$.
The right two diagrams hold for all
$f, f' : A \rightarrow B \in \ixcat_1$ and $g, g' : B \rightarrow C \in \ixcat_1$
and $\mathbf{f} : f \Rightarrow f' \in \ixcat_2$ and
$\mathbf{g} : g \Rightarrow g' \in
\ixcat_2$.
\end{definition}

\begin{proposition}
  Every pomonoid-graded monad is a 2-category-graded monad.
  \label{prop:gmonad-to-2-monadoid}
\end{proposition}

\begin{proof}
  Strict monoidal categories are equivalent to single-object
2-categories (Prop.~\ref{cor:monoidal-cats-as-single-cats}). Thus
  $\Gm : \mset{} \rightarrow [\basecat, \basecat]$ is equivalent to a
  lax functor $\T : {1}(\mset) \rightarrow \sing{\basecat}$,
  with $\T f = \Gm f$ (where $f$ is a morphism of ${1}(\mset)$
  and an object of $\mset$) and $\T \mathbf{f} = \Gm \mathbf{f}$ for approximations
  (where $\mathbf{f}$ is a 2-morphism of ${1}(\mset)$ and a
  morphism of $\mset$) and $\eta_{1} \, = \eta^{\Gm}$ and $\mu_{f,g}
  \, = \mu^{\Gm}_{f,g}$ whose lax functor axioms follow from the graded monad
  axioms.
\end{proof}

\section{Lax natural transformations and category-graded monad morphisms}
\label{sec:lax-nats}
Our use of lax functors provides a source of
useful results from the literature. Here we show the notion of
\emph{lax natural transformations}~(\cite{street1972two}) which provides
category-graded monad (homo)morphisms and further useful additional structure
leveraged in Section~\ref{sec:loose}.

\begin{definition}(\cite{street1972two})
Let $\ixcat$ be a category
and $\TS, \T : \ixcat \rightarrow \cats$ be two lax functors.
A \emph{left lax natural transformation} $L : \TS \xrightarrow{.} \T$
comprises (1) a functor $L_I : \TS I \rightarrow
\T I$ for every $I \in \objs{\ixcat}$ and (2)
a natural transformation $L_f : \T_f L_I \xrightarrow{.} L_J \TS_f$
for every $f : I \rightarrow J \in \morphs{\ixcat}$,
such that the following diagrams commute (where $f : I \rightarrow J$
and $g : J \rightarrow K$):
\begin{equation*}
\xymatrix@R=0.12em@C=3em{
\T_g \T_f L_I
 \ar[r]^-{\text{\footnotesize{$\mu^{\T}_{g,f} L_I$}}} \ar[dd]_-{\text{\footnotesize{$\T_g L_f$}}}
& \T_{g \circ f} L_I \ar[dr]^-{\text{\footnotesize{$ L_{g \circ f}$}}}   \\
& & L_K \TS_{g \circ f} \\
\T_g L_J \TS_{f} \ar[r]_-{\text{\footnotesize{$ L_g \TS_{f}$}}}
& L_K \TS_{g} \TS_{f} \ar[ur]_-{\text{\footnotesize{$ L_K \mu^{\TS}_{g,f}$}}}
}
\qquad\;\;
\xymatrix@R=1.75em@C=3em{
L_I \ar[r]^-{\text{\footnotesize{$\eta^{\T}_I L_I$}}} \ar[dr]_-{\text{\footnotesize{$ L_I \eta^{\TS}_{I}$}}}
 & \T_{id_I} L_I \ar[d]^-{\text{\footnotesize{$ L_{id_I}$}}} \\
& L_I \TS_{id_{I}}
}
\end{equation*}
A \emph{right lax natural transformation} $R : \TS \xrightarrow{.} \T$
has the same data, with functors $R_I : \TS I \rightarrow \T I$
but a family of natural transformations $R_f : R_J \TS_f \xrightarrow{.}
\T_f R_I$. Subsequently, the dual of the above diagrams commute.
\label{def:left-lax-nat-trans}
\end{definition}

\begin{proposition}
Let $\ixcat$ and $\basecat$ be categories and
$\TS, \T : \ixcat^\op \rightarrow \Endo{\basecat}$
be $\mathbb{I}$-graded monads.

A \emph{homomorphism} $\gamma : \T \rightarrow \TS$
between the $\mathbb{I}$-graded monads is a left-lax natural
transformation $L : \TS \xrightarrow{.} \T$ (since
$\TS$ and $\T$ map into a sub-category of
$\mathbf{Cat}$ with $\TS I = \T I = \basecat$) with
$L_I = \Id_{\basecat}$. Therefore we have a family of
natural transformations $\gamma_f = L_f : \T_f L_I \xrightarrow{.} L_J
\TS_f$ which are natural transformations $\T_f \xrightarrow{.}
\TS_f$ since $L_I$ and $L_J$ are the identity functor,
with the following homomorphism axioms following from
the definition of lax naturality (eliding again
$L_I$ and $L_J$ which are identities):
\begin{align*}
\xymatrix@R=0.1em@C=3em{
\T_f \circ \T_g
 \ar[r]^-{\text{\footnotesize{$\mu^{\T}_{f,g}$}}} \ar[dd]_-{\text{\footnotesize{$\T_f \gamma_g$}}}
& \T_{g \circ f} \ar[dr]^-{\text{\footnotesize{$\gamma_{g \circ f}$}}} \\
& & \TS_{g \circ f} \\
\T_f \circ \TS_{g} \ar[r]_-{\text{\footnotesize{$\gamma_f \, \TS_{g}$}}}
& \TS_{f} \TS_{g} \ar[ur]_-{\text{\footnotesize{$\mu^{\TS}_{f,g}$}}}
}
\qquad
\xymatrix@R=1.75em@C=3em{
\Id \ar[r]^-{\text{\footnotesize{$\eta^{\T}_I$}}} \ar[dr]_-{\text{\footnotesize{$\eta^{\TS}_{I}$}}}
 & \T_{\id_I} \ar[d]^-{\text{\footnotesize{$\gamma_{\id_I}$}}} \\
& \TS_{id_{I}}
}
\end{align*}
%
\label{def:monadoid-homomorphism}
Thus we can define the category of $\ixcat$-graded monads over
$\mathbb{C}$ base with these morphisms.
\end{proposition}

Lax natural transformations are also key to the next step of
capturing \emph{parameterised monads}.

\section{Parameterised monads and generalised units}
\label{sec:loose}
The notion of a monad with two indices which denote pre- and
post-conditions was first proposed for the continuation monad
by~\cite{wadler1994monads}. Later,~\cite{atkey2009parameterised}
generalised this idea, introducing the concept of a
\emph{parameterised monad} with a doubly-indexed functor
$\Pm : \ixcat{}^{\op} \times \ixcat{} \rightarrow [\basecat{},
\basecat{}]$ which
can, for example,
model effects with Floyd-Hoare-logic reasoning via indices of pre- and
post-conditions.

\begin{definition}
\label{def:parameterised-monad} (\cite{atkey2009parameterised})\footnote{
Atkey actually presents parameterised monads with a ternary functor
$\Pm : \ixcat^\op \times \ixcat \times \basecat
\rightarrow \basecat$;
we use an isomorphic binary functor here mapping to $[\basecat,\basecat]$
to reduce clutter and as this is akin to presentations of graded
monads on functors $\T : \ixcat^\op \rightarrow [\basecat, \basecat]$.}
A \emph{parameterised monad} comprises a
functor
 $\Pm : \ixcat^\op\times\ixcat\!\rightarrow\![\basecat,\!\basecat]$
and natural transformations
$\eta^{\Pm}_{I} : \Id_{\mathbb{C}} \xrightarrow{.} \Pm (I , I)$
and $\mu^{\Pm}_{I,J,K} : \Pm (I, J) \icirc \Pm (J, K) \xrightarrow{.} \Pm (I, K)$
satisfying analogous unit and associativity axioms to the usual monad
axioms (with the addition of the indexing).

Furthermore, $\eta^{\Pm}$ is dinatural in $I$ and
$\mu^{\Pm}_{I, J, K}$ is dinatural in $J$ and natural in $I$, $K$,
equating to the following dinaturality diagrams for all $f : I \rightarrow J$
and $g : J \rightarrow J'$:
%
\begin{align}
\xymatrix@C=6em@R=1.75em{
\Pm (I, J) \icirc\, \Pm(J', K) \ar[r]^{\text{\footnotesize{$\Pm (I, g) \, \Pm(J', K)$}}}
\ar[d]_{\text{\footnotesize{$\Pm (I, J) \Pm(g, K)$}}}
  & \Pm (I, J') \icirc\, \Pm(J', K) \ar[d]^{\text{\footnotesize{$\mu^\Pm_{I,J',K}$}}} \\
\Pm (I, J) \icirc\, \Pm (J, K) \ar[r]_-{\text{\footnotesize{$\mu^{\Pm}_{I,J,K}$}}} &  \Pm(I, K)
}
\qquad
\xymatrix@C=2.5em@R=1.75em{
\mathsf{\Id} \ar[r]^-{\text{\footnotesize{$\eta^{\Pm}_{I}$}}} \ar[d]_{\text{\footnotesize{$\eta^{\Pm}_{J}$}}}  & \Pm(I, I) \ar[d]^{\text{\footnotesize{$\Pm(I,f)$}}}  \\
\Pm(J, J) \ar[r]_{\text{\footnotesize{$\Pm(f, J)$}}}  & \Pm (I, J)
}
\label{eq:param-dinaturality}
\end{align}
Multiplication $\mu^{\Pm}$ requires that the post-condition of the outer
computation matches the pre-condition of the inner computation. Thus, similarly
to category-graded monads, parameterised monads restrict
the composition of computations. Where category-graded monads differ is that this
restriction is provided via morphisms, whereas the indices of parameterised
monads are just pairs. Thus, category-graded monads can have indices with more computational
content, e.g., proofs (see Example~\ref{exm:monadoid-restrict}).
\end{definition}

\begin{example}
\label{exm:param-typed-state}
\cite[\S2.3.2]{atkey2009parameterised} Mutable state can be modelled
by the state monad with $\T A = (A \times S)^S$ where $S$
represents the type of the state. A parameterised monad provides a
type-refined version of this where
$\Pm \, (S_1, S_2) \, A = (A \times S_2)^{S_1}$ modelling the type of
state at the start ($S_1$) and end of a computation ($S_2$).  The
parameterised monad operations have the same definition as the state
monad operations $\eta^\Pm_{S} = x \mapsto \lambda s . (x , s)$ and
$\mu^\Pm_{S_1,S_2,S_3} = c \mapsto \lambda s_1 . \textsf{app} \; (c \; s_1)$
where $c : (((A \times S_3)^{S_2}) \times S_2)^{S_1}$.  Reading and
writing from the store is provided by two families of operations:
$\textit{read}_{S} : \Pm (S, S) \, S$
 and $\textit{store}_{S_0,S} : S
\rightarrow \Pm (S_0, S) \, 1$.
\end{example}

This idea of using lax functors to generalise parameterised
monads has been conjectured before by~\cite{Capriotti} in a blog post,
but without further details.
Here we show that,
on there own, category-graded monads (lax functors) do not subsume parameterised monads, but
they instead capture a subset of parameterised monads restricted
to discrete indexing categories (with only identity morphisms). We
show that discrete parameterised monads are category-graded monads
(Section~\ref{sec:discrete}), and thus similar in power to
monoid-graded monads (i.e., without ordering). We then add the additional structure of
\emph{generalised units} to category-graded monads (which arises as a
kind of lax natural transformation) which accounts for the additional
power of full parameterised monads (Section~\ref{sec:full}).

\subsection{Discrete parameterised monads are category-graded monads}
\label{sec:discrete}

\newcommand{\dis}{\Delta} 
\newcommand{\ind}{\nabla} 
\newcommand{\discrete}[1]{\dis{(#1)}} 

Standard notions of \emph{discrete} and \emph{indiscrete} categories
are key here. We thus recall their definitions:

\begin{definition}
  A category is \emph{discrete} if its only morphisms are identities.

  The functor $\dis : \cats \rightarrow \cats$ discretises a
  category by discarding all but the identities
\end{definition}

\begin{definition}
   A category is \emph{indiscrete} if there is exactly
   one morphism between every pair of objects.

   The functor $\ind : \cats \rightarrow \cats$
   maps a category to its indiscrete form by
   replacing morphisms with pairs of the source and target objects (``dominoes''), i.e.,
   $\ind(\mathbb{C})(a, b) = \{(a, b)\}$.
  Identities are pairs of identical objects and
  composition of $(a, b)$ with $(b, c)$ yields $(a, c)$.
\label{def:cat-of-pairs}
\end{definition}

A standard result is that
the $\dis$ and $\ind$ functors arise from a single adjoint
triple (see Appendix, Prop.~\ref{prop:discrete-indiscrete}).
The symbol $\dis$ is often used for the diagonalisation functor,
and $\ind$ co-diagonalisation. We reuse the notation
as discretisation is akin to restricting a category to
the diagonal of the adjacency matrix of its morphisms.

\begin{definition}
A \emph{discrete parameterised monad} is a parameterised monad indexed by a
discrete category, with functor
$\Pm :\discrete{\ixcat}^\op \times \discrete{\ixcat} \rightarrow
[\basecat, \basecat]$ which has a degenerate morphism mapping
which is always the identity
$\Pm \, (\id, \id) = id : \Pm \, (I, I) \, A \rightarrow \Pm \, (I, I)
\, A$ since there are only identity morphisms in $\discrete{\ixcat}$.
Dinaturality squares (eq.~\ref{eq:param-dinaturality}) trivially
hold as they collapse to identities.
\end{definition}

\begin{proposition}
  Every discrete parameterised monad is a category-graded monad.
\label{prop:discrete-to-monadoid}
\end{proposition}

\begin{proof}
Let $(\Pm, \mu^{\Pm}, \eta^{\Pm})$ be a discrete parameterised monad on
$\Pm : \discrete{\ixcat}^{\op} \times \discrete{\ixcat} \rightarrow
[\basecat, \basecat]$. There is then a category-graded
monad $\T : \ind(\ixcat{})^\op \rightarrow \sing{\basecat}$, graded by
the indiscrete version of $\ixcat$ which has
pairs of objects as morphisms, and thus the morphism mapping of $\T$ is
defined $\T \, (I, J) = \Pm \, (I, J)$ with operations:
\begin{align*}
\eta_I  : \Id \xrightarrow{.} \T_{(I,I)} = \eta_I^{\Pm}
\qquad
\mu_{(I,J),(J,K)} : \T_{(I,J)} \icirc \T_{(J,K)} \xrightarrow{.}
                    \T_{(I,K)} = \mu^\Pm_{I,J,K}
\quad \textit{(for all $(I, J), (J, K) \in \morphs{\ind(\ixcat)}$)}
\end{align*}
This construction is injective, with an
inverse mapping from indiscrete-category-graded monads
 to discrete parameterised monads.
Each unique morphism of an indiscrete category $\jxcat{}$ can
be identified by its singleton homset $\jxcat{}(I, J)$.
Thus, there is a map from indiscrete-category-graded monads $\TS : \jxcat{}^\op \rightarrow \sing{\basecat}$ to discrete parameterised monads
$\Pm : \dis(\jxcat)^\op \times \dis(\jxcat) \rightarrow
[\basecat, \basecat]$ where $\Pm \, (I, J) \, A = \TS \; (\jxcat{}(I, J))
\, A$ and
$\Pm \, (I, J) \, f = \TS_{\jxcat{(I, J)}} f$, and operations
 $\eta^{\Pm}_{I} = \eta_{I}$ and $\mu^{\Pm}_{I,J,K} =
\mu_{\jxcat(I,J),\jxcat(J,K)}$.
\end{proof}


\subsection{Parameterised monads are category-graded monads with
a generalised unit}
\label{sec:full}

Of course, parameterised monads need not be discrete and therefore
may have non-degenerate morphism mappings
$\Pm \, (f, g) \, A : \Pm \, (I, J) \,A \rightarrow \Pm \, (I', J') \, A$
for $f : I' \rightarrow I$ and $g : J \rightarrow J'$ in $\ixcat$.
Through the Floyd-Hoare perspective, this morphism
mapping corresponds to pre-condition
strengthening (via $f$) and post-condition weakening (via $g$), i.e., a
kind of approximation on the indices.
Whilst object pairs $\ixcat^{\op} \times \ixcat$ do not
contain computational content in the way that a morphism may
(e.g., if it is a function), the morphism mapping provides a way
to change an effectful computation via the morphisms of $\ixcat{}$.

A category-graded monad could be constructed from a parameterised monad following
the approach for discrete parameterised category-graded monads with $\T : \ind(\ixcat{})^\op
\rightarrow \sing{\basecat}$ defined on morphisms as $\T \, (I, J) = \Pm \, (I,
J)$ (Prop~\ref{prop:discrete-to-monadoid}). However, such a
construction is non-injective since
the morphism mapping of $\Pm$ is discarded.
Alternatively, we might try to define $\T : \ixcat^\op \rightarrow
\sing{\basecat}$ with $\T (f : I \rightarrow J) = P (I, J)$.
The morphism mapping $\Pm (g, f) : \Pm (I, J) \xrightarrow{.} \Pm
(I', J')$ for $g : J \rightarrow J'$, $f : I' \rightarrow I$
 would then correspond to a family of natural transformations
of the form $\alpha_{f, g, k} : \T \, k \xrightarrow{.} \T (g \circ k \circ f)$
with $k : I \rightarrow J$.
However, no such family is elicited by a (2-)category-graded monad
alone. Furthermore, $\T (f : I \rightarrow J) = \Pm (I, J)$ is non-injective
since it maps to $\Pm (I, J)$ only when there is a morphism
$I \rightarrow J \in \morphs{\ixcat}$ despite the $\Pm$ object mapping
being defined for all object pairs $\ixcat^\op \times \ixcat$.

Additional structure is therefore needed to capture
full parameterised monads. For this, we introduce the notion
of a \emph{generalised unit}.

\begin{definition}
Let
$(\T : \ixcat^\op \rightarrow \Endo{\basecat}, \mu,
\eta)$ be an $\ixcat$-graded monad and
let $\mathbb{S}$ be a \emph{wide subcategory} $\mathbb{S} \subseteq \ixcat$ where
$\objs{\mathbb{S}} = \objs{\ixcat}$ (all the same objects) and $\iota :
\mathbb{S} \rightarrow \ixcat$ is the inclusion functor.

A \emph{generalised unit} augments a category-graded monad with
a right lax natural transformation (Def.~\ref{def:left-lax-nat-trans})
$R : \mathbf{Id} \xrightarrow{.} (\T \circ \iota)$
mapping from the identity category-graded monad (Example~\ref{exm:identity-monadoid}).
The right lax natural transformation
has $R_I = \Id_{\basecat}$ for all $I$
and thus its family of maps has signature $R_f : A \rightarrow \T_{\iota(f)} A$.
Subsequently, the right lax natural transformation axioms are
specialised to the following axioms, where we let $\geneta_f = R_f$
and make the inclusion functor $\iota$ implicit:
\newcommand{\hiota}[1]{#1}
\begin{equation*}
\xymatrix@R=2em@C=3em{
A \ar[d]_-{\text{\small{$\geneta_{g \circ f}$}}}
\ar[r]^-{\text{\small{$\geneta_f$}}} &
\T_{\hiota{f}}\, A
\ar[d]^-{\text{\small{$\T_{\hiota{f}} \, \geneta_g$}}} \\
\T_{\hiota{g \circ f}}\, A &
\ar[l]^-{\;\;\text{\small{$\mu_{f,g}$}}} \T_{\hiota{f}} \T_{\hiota{g}}\, A
}
\qquad
\xymatrix@R=2em@C=3em{
A \ar[r]^-{\text{\small{$\eta_I$}}} \ar[dr]_-{\text{\small{$\geneta_{\id_I, A}$}}}
 & \T_{\id_I}\, A \ar[d]^-{\text{\small{$\id_{\T_{\id_I A}}$}}} \\
 & \T_{\hiota \id_I}\, A
}
\end{equation*}
%
%
The left square states that
``generating'' two computations indexed by morphisms in
$\morphs{\mathbb{S}}$ via $\geneta$
and multiplying via $\mu$ is equivalent
to generating via $\geneta$ the computation indexed by the composition of
the indices.
The right square is well defined by the requirement that $\objs{\mathbb{S}} = \objs{\ixcat}$,
thus all identity morphisms of $\ixcat$ are in $\mathbb{S}$, and $\geneta_{\id_I}$
and $\eta_I$ coincide.

We refer to the family of maps $\geneta_{f, A} : A \rightarrow \T_f\, A$ as
the \emph{generalised unit}, with the notation alluding to $\eta$,
since it has a similar form to unit but is defined for all $f :
I \rightarrow J \in \morphs{\mathbb{S}}$ rather than just on identities.
\end{definition}

\begin{example}
  \label{prop:monadoid-as-loose}
  Every category-graded monad has a generalised unit
  with $\mathbb{S} = \discrete{\ixcat}$, i.e., the subcategory containing
only the identity morphisms. 
Since $\discrete{\ixcat}$ contains only the identity morphisms,
$\geneta$ need only be defined
on identities with
$\geneta_{id_I} = \eta_I : \Id_{\mathbb{C}} \xrightarrow{.}
\T_{id_I}$.
\end{example}

\begin{example}
Consider a pomonoid-graded monad $\T$ presented as a 2-category-graded monad
(Proposition~\ref{prop:gmonad-to-2-monadoid}) with $(\mset, \mmult,
\munit, \sqsubseteq)$ and $\T : 1(\mset) \rightarrow \sing{\basecat}$.
If the unit element is the bottom element of
the ordering, i.e., $\forall m \in \mset . \munit \sqsubseteq m$,
then there is a generalised unit with $\mathbb{S} = \mathbb{I}$
defined
$\geneta_m = \T(\munit \sqsubseteq m) \circ \eta : \mathsf{Id}_{\mathbb{C}}
\xrightarrow{.} \T_m$.
\end{example}

\newcommand{\prodext}[1]{{#1^{\ind}}}

\noindent
\begin{proposition}
  Every parameterised monad is a category-graded monad with a
  generalised unit.
\label{prop:param-monadoid-loose}
\end{proposition}

\begin{proof}
Let $(\Pm  : \ixcat^{\op} \times \ixcat \rightarrow [\basecat,
\basecat], \mu^{\Pm}, \eta^{\Pm})$ be a parameterised monad.
From $\ixcat$ we construct a category $\prodext{\ixcat}$
called the \emph{pair completion} of $\ixcat$ where
objects are the objects of $\ixcat$, and morphisms from $I$ to $J$ are
either a morphism in $\ixcat(I, J)$ or a pair $(I,J)$,  that is, we have homsets
$\prodext{\ixcat}(I, J) = \ixcat(I, J) \, \uplus \, \{(I, J)\}$.
%
Thus, the homsets are the disjoint union of $\ixcat$ morphisms
or a pair $(I, J)$, with injections into it written $\inj_1$ and
$\inj_2$. Identities of $\prodext{\ixcat}$ are by $\inj_1(id_I)$ and composition is
defined:
\begin{equation}
\label{eq:prodext-comp}
(g : J \rightarrow K) \circ (f : I \rightarrow J) = \begin{cases}
\inj_1 (g' \circ f') & f = \inj_1(f') \; \wedge \; g = \inj_1(g') \\
\inj_2 (I, K) & \textit{otherwise}
\end{cases}
\end{equation}
We then define a category-graded monad
$\T : (\prodext{\ixcat})^\op \rightarrow \sing{\basecat}$ where $\T \, (f : I \rightarrow J) = \Pm \; (I, J)$
with operations:
$\eta_{I} = \eta^{\Pm}_{I}$
and
$\mu_{f,g} = \mu^{\Pm}_{I,J,K}$ for $f : I \rightarrow
  J, g : J \rightarrow K \in \morphs{\prodext{\ixcat}}$.
The source and target objects $I$, $J$, and $K$
are used to index $\mu^\Pm$ without needing to determine
whether the morphisms $f$ and $g$ are in the left or right injection
 of $\prodext{\ixcat}$.

Let $\mathbb{S} =
\ixcat$ which is a subcategory of $\prodext{\ixcat}$ via the inclusion which is the identity on objects and
left injection $\inj_1$ on morphisms. This satisfies
the property that all identities of $\prodext{\ixcat}$ are in
$\mathbb{S}$ as identities are given by $\inj_1 \, id$.
Generalised unit $\geneta$ then has two equivalent definitions:
\begin{equation}
\forall f : I \rightarrow J \in \morphs{\mathbb{\ixcat}} \, . \;\;\;
\geneta_{f} \, = \, \Pm \, id_I \, f \circ \eta^P_I \, = \, \Pm \, f \, id_J \circ \eta^P_J
\label{eq:geneta-functor-morphism}
\end{equation}
These two definitions
of $\geneta$ are equivalent
by dinaturality of $\eta^{\Pm}$
(the right equality \emph{is}
the dinaturality condition). The generalised unit
axioms follow from bifunctor axiom $\Pm \, id_I \, id_I = id$
 and dinaturality of $\mu^{\Pm}$.

 In the Appendix,
 Proposition~\ref{prop:parameterised-monad-result-inverse} shows that
 the above construction has a left inverse, i.e., \emph{every
   category-graded monad with generalised unit indexed by
   $\prodext{\ixcat}$ with $\mathbb{S} = \ixcat$ has a corresponding
   parameterised monad}. We give an outline here.

For a category-graded monad on
 $\T : (\prodext{\ixcat})^\op \rightarrow \sing{\basecat}$
with $\geneta$
 and subcategory $\mathbb{S} = \ixcat$
we can construct a parameterised
 monad with functor $\Pm : \ixcat^{\op}
  \times \ixcat \rightarrow [\basecat, \basecat]$
on objects as $\Pm \, (I, I) = \T \; (\inj_1 \id_I)$ and
 $\Pm \, (I, J) = \T \; (\inj_2 \; (I, J))$ (for $I \neq J$).
The morphism mapping $\Pm (f, g) \, h$
is built from $\mu$ and $\geneta$,
where for all $f : I' \rightarrow I, g : J \rightarrow J'$ $\in
   \morphs{\ixcat}$, $h : A \rightarrow B \in \morphs{\basecat}$
 (with shorthand $k = \inj_2 (I, J)$, $\bar{f} =
 \inj_1 \, f$ and $\bar{g} = \inj_1 \, g$):
 \[
 \hspace{-0.3em}
 \xymatrix@C=1.3em@R=0.5em{
 \T_k \, A
 \ar[rr]^-{\text{\footnotesize{$\geneta_{\bar{f}}\T_k\;$}}}  & &
 \T_{\bar{f}} \T_k \ar[rr]^-{\text{\footnotesize{$\T_{\bar{f}} \, \T_k \,
   \geneta_{\bar{g}}\;\;$}}} \, A
 & &
 \T_{\bar{f}} \T_k \T_{\bar{g}} \, A
 \ar[rr]^-{\text{\footnotesize{$\mu_{\bar{f},k,\T {\bar{g}}}$}}} & &
 \T_{(k \circ \bar{f})} \T_{\bar{g}} \, A
 \ar[rr]^-{\text{\footnotesize{$\mu_{k \circ \bar{f},\bar{g}}$}}}
  & &
 \T_{(\bar{g} \circ k \circ \bar{f})} \, A
 \ar[rr]^-{\text{\footnotesize{$\T h$}}}
 & &
 \T_{(\bar{g} \circ k \circ \bar{f})} \, B
 }
 \]
Bifunctoriality of $\Pm$ follows
 from the right lax natural transformation
 and category-graded-monads axioms.
 The parameterised monad operations are provided
 by category-graded monad $\mu^{\Pm}_{I,J,K} =
 \mu_{\inj_2 (I,J), \inj_2 (J,K)}$ and generalised unit $\eta^{\Pm}_{I} =
 \geneta_{\inj_1 \id_I}$.

 This construction is the inverse of the former, thus
 there is just one category-graded monad with
 generalised unit for every parameterised monad.
 \end{proof}

 \begin{corollary}
   \label{cor:final-result}
  2-cat-graded monads with a generalised unit
   subsume graded and parameterised monads.
\end{corollary}

The mapping of a parameterised monad to our structure here
identifies two parts of a parameterised monad:
the object mapping of $\Pm$, defined for all pairs of objects and the
morphism mapping of $\Pm$, defined for all morphisms. These two classes
are grouped into a single category via $\prodext{\ixcat}$
such that the generalised unit is defined only on the actual morphisms
of $\ixcat$, corresponding to the morphism mapping part of $\Pm$.

\begin{remark}
\citet{atkey2009parameterised}
  describes parameterised monads in a similar way to our
  description of category-graded monads: ``\emph{parameterised monads are to monads as categories are to
    monoids}''. He illustrates
  this by generalising the {writer monad}
  $\T A = M \times A$ for some monoid on $M$, replacing
  $M$ with morphisms of a small category. That is,
  for some small category $S$ then
  $\Pm : S^{\op} \times S \rightarrow [\mathbf{Set}, \mathbf{Set}]$ is
  defined $\Pm (I, J) \, A = S(I, J) \times A$, i.e., the
  set of $I \rightarrow J$ morphisms paired with the set $A$. Then
  $\eta^\Pm_{I,A} a \mapsto (id_I, a)$ and
  $\mu^{\Pm}_{I,J,K,A} (f, (g, a)) \mapsto (g \circ f, a)$. This
  construction is essentially a value-level version of what the
  indices of a category-graded monad provide. Thus category-graded monads
  provide a static trace of morphism composition (recall
  Example~\ref{exm:static-trace}),
  whilst parameterised monads can give only a dynamic, value-level trace.
\end{remark}

\begin{example}\emph{Making parameterised monads constructive via a
    category-graded monad.}
  We define a class of category-graded monads based on
parameterised monads, but that are \emph{not} parameterised monads. 
%
A parameterised monad
$(\Pm : {\ixcat}^{\op} \times {\ixcat} \rightarrow [\basecat,
\basecat], \mu^{\Pm}, \eta^{\Pm})$
induces a cat-graded monad on $\T : \ixcat^\op \rightarrow \sing{\basecat}$
with $\T \, (f : I \rightarrow J) = \Pm \; (I, J)$,
i.e., source and target objects of $f$ provide the
parameterised monad indices with operations
$\eta_{I} = \eta^{\Pm}_{I}$ and
$\mu_{f,g} = \mu^{\Pm}_{I,J,K}$ (for $f : I
\rightarrow J$, $g : J \rightarrow K$) and
a generalised unit
$\geneta_f = \Pm \, \id_I \, f \circ \eta^{\Pm}_I$.
This gives a restricted view of the
parameterised monad $\Pm$, allowing computations
$\Pm \; (I, J)$ to be used only when there is a morphism (e.g., a proof
or ``path'') $I \rightarrow J \in \morphs{\ixcat}$.
We thus call this a \emph{constructive} parameterised monad.
This example is only possible
with the additional power of category-graded monads.
\label{exm:monadoid-restrict}
\end{example}

\section{Example: combining graded monads and parameterised monads}
\label{sec:combined-example}
\newcommand{\ubjudg}[4]{\vdash_{#1} #2 : #3 \Rightarrow #4}

Since 2-category-graded monads with generalised units
unify both parameterised and
graded monads (Corollary~\ref{cor:final-result})
then they can provide a model for systems which
combine both quantitative reasoning and program logics into a single
structure. For example,
\cite{DBLP:conf/icalp/BartheGGHS16} define a probabilistic Hoare Logic
(\textsf{aHL}) which provides a quantitative analysis of the ``union
bound'' of probabilistic computations.  Judgments in \textsf{aHL} for
a program $c$ are of the form: $\ubjudg{\beta}{c}{\phi}{\psi}$ where
the initial state of a program $c$ satisfies pre-condition $\phi$ and
after execution produces a distribution of states for which $\psi$
holds. The annotation $\beta$ is the maximum probability that $\psi$
does not hold. Derivations of judgments track the change in this
probability bound $\beta$ by structure on the annotation (where
$\beta, \beta' \in [0, 1]$ and $+$ is saturating at $1$):
\begin{align*}
\hspace{-2em}\begin{array}{c}
\dfrac{}{\ubjudg{0}{\mathbf{skip}}{\phi}{\phi}} \trule{skip}
\quad
\dfrac{}{\ubjudg{0}{x \leftarrow e}{\phi[e/x]}{\phi}} \trule{assgn}
\quad
\dfrac{\ubjudg{\beta}{c}{\phi}{\phi'} \;\;
  \ubjudg{\beta'}{c'}{\phi'}{\phi''}}{
  \ubjudg{\beta + \beta'}{c; c'}{\phi}{\phi''}} \trule{seq}
\\[1.5em]
\dfrac{\models \phi' \rightarrow \phi \quad \vdash_{\beta} c :
\phi \Rightarrow \psi \quad \models \psi \rightarrow \psi'
\quad \beta \leq \beta'}{\vdash_{\beta'} c : \phi' \Rightarrow
\psi'}\trule{weak}
\end{array}
\end{align*}
The consequence rule (called \emph{weak} above) combines pre-condition
strengthening and post-condition weakening with approximation of the
probability upper-bound.

We give a 2-category-graded monadic semantics to this system, using a product of two
indexing categories: the 2-category of the additive
monoid over $[0, 1]$ (with ordering $\leq$) and the category $\prodext{\mathbf{Prop}}$
whose morphisms are pairs of
propositions $\mathbf{Prop}$ and logical implications.
For brevity we consider just the
subset of the \textsf{aHL} system given above, though the rest of
the system can be readily modelled.

Barthe \emph{et al.} interpret programs as
functions from a store type $\mathsf{State}$ to distributions over
stores $\textbf{Distr}(\mathsf{State})$. We extend this to
include a return value of type $A$ denoting
$\mathsf{D}_A = \mathsf{State} \rightarrow
\textbf{Distr}(\mathsf{State} \times A)$.
We define a 2-category-graded monad on $\mathbf{Set}$
as a lax functor $\T : ([0, 1] \times \prodext{\mathbf{Prop}})^{\mathsf{op}}
\rightarrow [\mathbf{Set}, \mathbf{Set}]$ with
generalised unit, combining Atkey's parameterised monad for a program logic (which
refines the set of denotations to valid store
transformations~\cite[p.27]{atkey2009parameterised}) with the
additional validity requirements with regards probabilities in
\textsf{aHL} (the probability that the return state does not satisfy
the post-condition is bounded above by $\beta$) where:
\newcommand{\cdomain}[1]{\textsf{St} \rightarrow \textbf{Dist}(#1 \times \textsf{St})}
  \begin{equation*}
    \T_{(\beta, f : \phi \rightarrow \psi)} A =
  \{ c \in \mathsf{D}_A \mid \forall s_1 .\ s_1 \models \phi \Rightarrow
  \exists (s_2, a) . ((s_2, a) \in c(s_1) \; \wedge \; s_2 \models \psi \; \wedge \;
  \underset{s_2}{\text{Pr}}[\neg \psi] \leq \beta) \}
\end{equation*}
The definition of the category-graded monad operations is essentially that of a
state monad combined with a distribution monad (e.g.,
$\eta_{(0, \phi),A} : A \rightarrow \T_{0, id_\phi : \phi \rightarrow
  \phi} A = x \mapsto \lambda s . \lambda p . (s, x)$) but whose set
of values is refined by the validity requirements of the above
definition. The generalised unit operation is
defined when $\models \phi' \rightarrow \phi$ then
$\geneta_{(0, f : \phi' \rightarrow \phi)} : A
\rightarrow \T_{(0,f)} A = x \mapsto \lambda s . \lambda p . (s, x)$
 where $f$ is the model of the implication in
 $\textbf{Prop}$.
Derivations are then interpreted as morphisms
$\interp{\vdash_\beta c : \phi \Rightarrow \psi} : 1 \rightarrow \T_{\beta, f : \phi \rightarrow \psi} 1$ with:
\begin{equation*}
\interp{\textbf{skip}} = \eta_{(0,\phi),1} : 1 \rightarrow \T_{0, id_\phi :
  \phi \Rightarrow \phi} 1
\qquad
\interp{c; c'} = \mu_{(\beta, f), (\beta', g), 1} \circ \T_{\beta, f :
  \phi \rightarrow \phi'} \interp{c'} \circ \interp{c}
\end{equation*}
The (weak) rule is modelled by generalised unit
and the approximation maps of the 2-cat-graded monad:
\begin{equation*}
\mu_{(\beta', g \circ f), (0, g')} \circ \T_{g \circ f} \geneta_{(0, g' : \phi' \rightarrow \phi)} \circ \mu_{(0, g), (\beta', f)} \circ \geneta_{(0, g : \psi \rightarrow \psi')} \circ \T_{\mathbf{f} : \beta \leq \beta'} \circ (\interp{c} : 1 \rightarrow \T_{\beta, f : \phi \rightarrow \psi} 1)
\end{equation*}
%


\section{Discussion}
\label{sec:discussion}
\subsection{A more direct relationship between parameterised and
  graded monads}

\noindent
The main motivation for category-graded monads is to unify graded and
parameterised monads. However, in some restricted cases, some
parameterised monads can be mapped to graded monads more directly.

\begin{proposition}
For a parameterised monad $(\Pm : \ixcat^\op \times \ixcat\!\rightarrow\![\basecat,\!\basecat], \mu^{\Pm}, \eta^{\Pm})$ where $\ixcat$ is a monoidal category $(\ixcat, \bullet, e)$
and $\basecat$ is \emph{finitely complete} (in that it has all finite limits)
there is a graded monad given by an \emph{end} in $\basecat$ (generalising universal
quantification) $\Gm \, f \, = \int_{i} \Pm \, (i, i \bullet f)$.
with unit
${\eta} : A \rightarrow \Gm {e} = \eta^{\Pm} : A \rightarrow \Pm (i, i \bullet e)$ and
graded monad multiplication as follows (which shows some
calculation):
{\small{
\begin{align*}
\Gm f \, \Gm g =\!\!\! \int_{i} \Pm (i, i \bullet f) \!\! \int_{j}\!\! \Pm (j, j \bullet g)
=\!\!\! \int_{i} \Pm (i, i \bullet f) \Pm (i \bullet f, (i \bullet f)  \bullet g) \xrightarrow{\mu^\Pm}\!\!\!\int_{i} \Pm (i, (i \bullet f) \bullet g) =\!\!\! \int_{i} \Pm (i, i \bullet (f \bullet g)) = \Gm (f \bullet g)
\end{align*}
}}
\end{proposition}
\noindent
Mapping the other way (graded into parameterised) is more
difficult as two indices are needed from one.

\subsection{Categorical semantics}
Given a notion of tensorial strength for category-graded monads, it is
straightforward to define a calculus whose denotational model is given by
the category-graded monad operations, i.e., a language
like the monadic metalanguage of~\citet{moggi1991monads} (akin to
Haskell's \emph{do}-notation) for sequential composition of effectful
computations. The core sequential composition and identity rules are
then of the form:
\begin{align*}
\dfrac{
  \Gamma \vdash t : \T_f \, A
\quad\, \Gamma, x : A \vdash t' : \T_g \, B
\quad\, f : I \rightarrow J
\quad\, g : J \rightarrow K}
{\Gamma \vdash \synLet{x}{t}{t'} : \T_{g \circ f} \, B}
\quad
%
%
\dfrac{\Gamma \vdash t : A}
      {\Gamma \vdash \langle{t}\rangle : \T_{id_I} \, A}
\end{align*}
The semantics of such a calculus resembles that of the monadic
meta language, taking the same form but with the added morphism grades. This
semantics then requires a notion of tensorial
strength in the category-graded setting, that is
a natural transformation
$\sigma_{f,A,B} : A \times \T_f \, B \rightarrow \T_f \, (A \times B)$
for all $f : I \rightarrow J \in
\morphs{\ixcat{}}$ satisfying a graded variant of the usual
monadic strength axioms.
Tensorial strength has been previously
considered for graded \citep{katsumata2014parametric} and
parameterised monads \citep{atkey2009parameterised}. Defining
a subsuming notion of \emph{strong category-graded monads}
as above appears to be straightforward.

\subsection{Further work}
\label{subsec:oidification-alt}


\paragraph{Category-graded comonads}
Various works employ \emph{graded comonads} to give the semantics
of \emph{coeffects}
(\cite{petricek2013coeffects,gaboardi2014,ghica2014}).
Category-graded monads dualise
straightforwardly to \emph{category-graded comonads} as a \emph{colax} functor $\D : \ixcat{} \rightarrow
\sing{\basecat{}}$ witnessed by natural transformations
$\varepsilon_{I} : \D_{\id_I} \xrightarrow{.} \Id$
and $\delta_{f,g} : \D_{g \circ f} \xrightarrow{.} \D_g \D_f$
with dual axioms to category-graded monads. The source of the
colax functor is dual to the $\ixcat^\op$ source of a category-graded
monad lax functor.

Graded comonads are category-graded comonads
however graded comonads usually have a
semiring structure on their indices, adding extra graded monoidal
structure (called \emph{exponential graded
  comonads}~\citep{gaboardi2016}). Further work is to incorporate such
additional structure into our formulation.

\paragraph{Polymonads, supermonads, and productoids}
\emph{Polymonads} provide a programming-oriented generalisation of
monads replacing the single type constructor $M$ of a monad with a
family of constructors
$\Sigma$~(\cite{hicks2014polymonadic}). For some
triples of $\textsf{M},\textsf{N},\textsf{P} \in \Sigma$ there is a
\emph{bind} operation of type:
$
\forall a, b . \sf{M} \, a \rightarrow
(a \rightarrow \sf{N} \, b) \rightarrow \sf{P}
\, b
$
which allows two effectful computations captured by $\sf{M}$
and $\sf{N}$ to be composed and encoded by $\sf{P}$. This generalises the
familiar \emph{bind} operation for monadic programming
\iflong
 of type
$\forall a, b . \sf{M} \, a \rightarrow (a \rightarrow \sf{M} \, b)
\rightarrow \sf{M} \, b$
\fi
which internalises Kleisli extension of a strong monad.

\iflong
\citet{bracker2016supermonads} provide a similar structure to
polymonads, called \emph{supermonads} where the constructors
$\M, \N, \Pm$ are derived from an indexed family.
The additional
power of supermonads is that each functor need not be an endofunctor,
but instead maps from some subcategory of the base category. This
provides models which have some type-predicate-like
restrictions on functors.
\fi

Another related structure is the
\emph{productoid} of~\citet{tate2013sequential} with an
indexed family of constructors $\T : \textit{Eff} \rightarrow
[\basecat, \basecat]$, indexed by an \emph{effectoid} structure \textit{Eff}
which is a kind of relational ordered monoid. Tate mentions the
possibility of modelling this 2-categorically.

Future work is to study a generalisation of these seemingly
related structures. 

\paragrapho{Hoare and Dijkstra Monads} \citet{nanevski2008hoare} introduced the notion of \emph{Hoare monads}
which resemble parameterised monads but indexed by pre-conditions
that are dependent on a heap value and post-conditions which are
dependent on a heap and the return value of a computation.  This
is generalised by the notion of \emph{Dijkstra}
monads~\citep{swamy2013verifying,maillard2019dijkstra}) indexed by a predicate
transformer which, given a post-condition on the final heap and
value, computes the weakest pre-condition of the computation.
Studying how these structures fit with other kinds of indexed
generalisations of monads, perhaps through the lense of lax functors, is interesting further
work.

\paragraph{Monads on indexed sets}
An alternate indexed generalisation of monads is to
consider monads over indexed sets, as in the work
of~\cite{mcbride2011functional}. This provides a fine-grained model of
effects which can react to external uncertainty in the program state
in a more natural way than the approaches discussed here. Exploring
the connection between monads on indexed sets/types and
cat-graded monads is further work.


\paragraph{The alternate oidification} Category-graded monads were presented as the oidification
of monads as a lax functor $\T : \mathbf{1} \rightarrow \sing{\basecat}$ to
$\T : \ixcat^\op \rightarrow \sing{\basecat}$. An
alternate, orthogonal oidification
is to generalise the target category as well
to the lax functor $\T : \ixcat^\op \rightarrow \cats$
where each $\T (f : I \rightarrow J)$ is then a functor from category
$\T I$ to $\T J$, rather than an endofunctor on a particular category.

Such a generalisation suggests further control over the models of
effects where for a morphism $f : I \rightarrow J$ the categories
$\T I$ and $\T J$ may differ. For example, they may be subcategories
of some overall base category as in the supermonads of
\citet{bracker2016supermonads}. Exploring this is future work.

\paragraph{Adjunctions}
\citet{fujii2016towards} show how a graded monad can be factored into
an adjunction with graded analogues of the Kleisli and Eilenberg-Moore
constructions. Relatedly, \cite{atkey2009algebras} gives notions of
Kleisli and Eilenberg-Moore category for parameterised
monads. Defining the appropriate adjunctions which gives
rise to category-graded monads is further work. This would provide the opportunity
to consider different kinds of semantics, for example,
for call-by-push-value (\cite{levy2006call}).

\subsection{Summary and concluding remarks}

Following B\'{e}nabou's perspective of
monads as degenerate lax functors $\T : \mathbf{1} \rightarrow
\sing{\basecat}$ involving single object
categories, we applied the notion of oidification, generalising
the source category from a point to a category  $\T : \ixcat^\op
\rightarrow \sing{\basecat}$. This yielded the
base notion of category-graded monads from which we
encompass the full power of graded and parameterised monads
as found in the recent literature. Namely:
\begin{itemize}[itemsep=0pt]
\item Monoid-graded monads are category-graded monads
(Prop.~\ref{prop:gmonad-to-monadoid});
\item Pomonoid-graded monads are 2-category graded monads
(Prop.~\ref{prop:gmonad-to-2-monadoid});
\item Discrete parameterised monads are category-graded monads
  (Prop.~\ref{prop:discrete-to-monadoid});
  \item Parameterised monads are category-graded monads with a
  generalised unit (Prop.~\ref{prop:param-monadoid-loose}) $\geneta_f : \Id \xrightarrow{.} \T_f$ (a right lax
  natural transformation) for all $f \in \morphs{\mathbb{S}}$ where
  $\mathbb{S} \subseteq \mathbb{I}$ and
  $\objs{\mathbb{S}} = \objs{\mathbb{I}}$.
\end{itemize}
Each result is an injection into
the more general structure. This could be made stronger
by a full and faithful embedding, but this requires
a definition of categories of parameterised monads and graded
monads, but the former is not provided in the literature.
This is future work.

Category-graded monads have the feel of a \emph{proof relevant}
version of a graded monad or parameterised monad where the index is
not merely a value, but a computation, that is a program or a proof.
We leave the door open for future applications to utilise this
generality, not just to provide a framework for including both graded
and parameterised monads into one cohesive whole, but also for programs
and models with more fine-grained computational indices.

\paragrapho{Acknowledgments} Thanks to Bob Atkey, Iavor Diatchki,
Shin-ya Katsumata, and Ben Moon for useful discussions and comments on
a draft of this paper. Specific thanks to Iavor for pointing out the
connection to the Grothendieck construction and Shin-ya for various
mathematical insights. We also thank the reviewers for their helpful feedback.

\bibliography{references}

\appendix

\section{Additional background and details}
\label{app:additional-background}

\begin{definition}
 \emph{2-categories} extend the notion of a category with
morphism between morphisms.
A 2-category $\mathbb{C}$ has a class of objects $\objs{\mathbb{C}}$, a
class of 1-morphisms (usual morphisms, between objects)
$\morphs{\mathbb{C}}$,
and a class of 2-morphisms (between 1-morphisms)
$\twomorphs{\mathbb{C}}$.
We write 2-morphisms in bold, e.g. $\mathbf{k} : f \Rightarrow g$
is a 2-morphism between 1-morphisms $f, g : A \rightarrow B$.

1-morphisms compose as usual via $\circ$ and have identities $id_A$
for all objects $A$. 2-morphisms have two notions of composition:
\emph{horizontal}, written $\circ_0$, which composes along objects and
\emph{vertical} written $\circ_1$ which composes along morphisms.
That is, horizontal composition of a 2-morphism $\textbf{i} : f \Rightarrow g$
where $f, g : A \rightarrow B$ and $\textbf{j} : f' \Rightarrow g'$
where $f', g' : B \rightarrow C$ yields
$\textbf{j} \circ_0 \textbf{i} : (f' \circ f) \Rightarrow (g' \circ
g)$. For morphisms $f, g, h : A \rightarrow B$ and
2-morphisms $\textbf{i} : f
\Rightarrow g$ and $\textbf{j} : g \Rightarrow h$ then
their vertical composition is $\textbf{j}
\circ_1 \textbf{i} : f \Rightarrow h$.

Both vertical and horizontal composition are associative
and have an identity via the identity 2-morphism $\textbf{id}_f : f \Rightarrow f$.
Additionally, vertical and horizontal composition satisfy the
\emph{interchange} axiom:
$(\textbf{i} \circ_1 \textbf{j}) \circ_0 (\textbf{k} \circ_1
\textbf{l})
= (\textbf{i} \circ_0 \textbf{k}) \circ_1 (\textbf{j} \circ_0
\textbf{l})$.
\label{def:2-cat}
\end{definition}

\begin{proposition}
  \label{prop:discrete-indiscrete}
  The following adjoint triple gives rise to functors for mapping
  (small) categories to their discrete and indiscrete versions:
  \begin{align*}
    \mathsf{d} \dashv \mathsf{ob} \dashv \mathsf{i} : \mathbf{Set} \rightarrow \cats
  \end{align*}
  where $\mathsf{d} : \mathbf{Set} \rightarrow \cats$ maps a set to a
  discrete category (the set gives the objects of the category),
  $\mathsf{ob} : \cats \rightarrow \mathbf{Set}$ maps a category
  to its set of objects, $\mathsf{i} : \mathbf{Set} \rightarrow \cats$
  maps a set to an indiscrete category with morphisms given
  by the unique pair of the objects, i.e., $(\mathsf{i} A)(a, b) = \{(a, b)\}$.
  Identities are pairs of identical objects and
  composition of $(a, b)$ with $(b, c)$ yields $(a, c)$; think dominoes.

  From this adjoint triple there arises two functors
  $\Delta = \mathsf{d} \circ \mathsf{ob} : \cats \rightarrow \cats$
  which discretises categories by discarding all but the identities
  and $\nabla = \mathsf{i} \circ \mathsf{ob} : \cats \rightarrow
  \cats$ which maps a category to its indiscrete form
  by replacing their morphisms with pairs of objects.
\end{proposition}

The following provides the details
of the inverse direction of Proposition~\ref{prop:param-monadoid-loose}.

\begin{proposition}
\label{prop:parameterised-monad-result-inverse}
Every category-graded monad with generalised unit indexed by
$\prodext{\ixcat}$ with $\mathbb{S} = \ixcat$ has a corresponding
 parameterised monad.
\end{proposition}

\newcommand{\plabel}[1]{\{\footnotesize{\textit{#1}}\}}
\begin{proof}
Let $(\T, \eta, \mu, \geneta)$ be a category-graded monad with
generalised unit $\geneta$ and
 $\T : (\prodext{\ixcat})^\op \rightarrow \sing{\basecat}$
 and subcategory $\mathbb{S} = \ixcat$. Then, there is a parameterised
 monad with $\Pm : \ixcat^{\mathsf{op}}
  \times \ixcat \rightarrow [\basecat, \basecat]$
 defined on objects as $\Pm \, (I, I) = \T \; (\inj_1 \id_I)$ and
 $\Pm \, (I, J) = \T \; (\inj_2 \; (I, J))$ (for $I \neq J$)
 and on morphisms $f : I' \rightarrow I, g : J \rightarrow J'$ $\in
   \morphs{\ixcat}$, $h : A \rightarrow B \in \morphs{\basecat}$
 (with shorthand $k = \inj_2 (I, J)$, $\bar{f} =
 \inj_1 \, f$ and $\bar{g} = \inj_1 \, g$):
 \begin{equation}
 \!\!\!\!\!\!\!\Pm (f, g) \, h =\T_{\bar{g} \circ k \circ \bar{f}} h \,\circ\,
 \mu_{k \circ \bar{f}, \bar{g}}
 \,\circ\, \mu_{\bar{f},k}\T_{\bar{g}} \,\circ\, \T_{\bar{f}} \, \T_k \,
 \geneta_{\bar{g}} \,\circ\,
 \geneta_{\bar{f}}\T_k \;\; : \, \Pm \, (I, J) \, A \rightarrow \Pm \, (I', J') \,
   B
 \label{eq:functor-converse-geneta}
 \end{equation}
 The following shows part of this definition
 diagrammatically for clarity:
 \[
 \hspace{-0.3em}
 \xymatrix@C=1.3em@R=0.5em{
 \T_k A
 \ar[rr]^-{\geneta_{\bar{f}}\T_k\;}  & &
 \T_{\bar{f}} \T_k \ar[rr]^-{\T_{\bar{f}} \, \T_k \,
   \geneta_{\bar{g}}} A
 & &
 \T_{\bar{f}} \T_k \T_{\bar{g}} A
 \ar[rr]^-{\mu_{\bar{f},k,\T \bar{g}}} & &
 \T_{(k \circ \bar{f})} \T_{\bar{g}} A
 \ar[rr]^-{\mu_{k \circ \bar{f},\bar{g}}}
  & &
 \T_{\bar{g} \circ k \circ \bar{f}} A
 \ar[r]^-{\T h}
 &
 \T_{\bar{g} \circ k \circ \bar{f}} B
 }
 \]
 The morphism mapping of $\Pm$ is thus built from $\mu$ and $\geneta$
 and is well defined since the
 functor $\T (\bar{g} \circ (\inj_2 (I, J)) \circ \bar{f})$ equals
  $\T (\inj_2 (I', J'))$ by composition in
 $\prodext{\ixcat}$ (eq.~\ref{eq:prodext-comp}).
 Thus, the domain of the functor matches the object
 mapping of $\Pm (I', J')$. Since $\mathbb{S} = \ixcat$,
  $\geneta$ is defined for all morphisms of $\ixcat$  which
 are wrapped by $\inj_1$ to indicate they are
 morphisms on $\ixcat$ rather than object pairs.

 Bifunctoriality of $\Pm (f, g) h$ holds as follows:
\begin{itemize}
\item $\Pm (f' \circ f, g' \circ g) (h' \circ h)
  = \Pm (f, g') h' \circ \Pm (f', g) h\ :\
\Pm (K', I) A \rightarrow \Pm (I', K) C$

for all $h : A \rightarrow B, h' : B \rightarrow C
\in \morphs{\basecat}$ and $g : I \rightarrow J,
g' : J \rightarrow K \in \morphs{\ixcat}$
and $f : I' \rightarrow J'$, $f' : J' \rightarrow K'
\in \morphs{\ixcat}$

follows from
  the right lax natural transformation axiom on $\mu$ and associativity
  of $\mu$, via the following reasoning
(where $k = \inj_2 (K, I')$
and $k' = \inj_2 (J', J) = \bar{g} \circ k \circ \bar{f'}$
by the definition of composition for $\prodext{\ixcat}$
(eq.\ref{eq:prodext-comp}),
and a line above a morphism means that $\inj_1$ is applied to it):

{\scalebox{0.71}{
\hspace{-1em}
\begin{minipage}{1\linewidth}
\begin{align*}
& \Pm (f, g') h' \circ \Pm (f', g) h \\
\equiv \; &
\T_{\bar{g'} \circ k' \circ \bar{f}} h' \,\circ\,
 \mu_{k' \circ \bar{f}, \bar{g'}}
 \,\circ\, \mu_{\bar{f},k'}\T_{\bar{g'}} \,\circ\, \T_{\bar{f}} \, \T_{k'} \,
 \geneta_{\bar{g'}} \,\circ\,
 \geneta_{\bar{f}}\T_{k'}
\circ
\T_{\bar{g} \circ k \circ \bar{f'}} h \,\circ\,
 \mu_{k \circ \bar{f'}, \bar{g}}
 \,\circ\, \mu_{\bar{f'},k}\T_{\bar{g}} \,\circ\, \T_{\bar{f'}} \, \T_k \,
 \geneta_{\bar{g}} \,\circ\,
 \geneta_{\bar{f'}}\T_k \\
\plabel{$\geneta$ nat.}
%
%
\equiv \; &
\T_{\bar{g'} \circ k' \circ \bar{f}} h' \,\circ\,
 \mu_{k' \circ \bar{f}, \bar{g'}}
 \,\circ\, \mu_{\bar{f},k'}\T_{\bar{g'}} \,\circ\, \T_{\bar{f}} \, \T_{k'} \,
 \geneta_{\bar{g'}} \,\circ\,
\T_{\bar{f}} \T_{\bar{g} \circ k \circ \bar{f'}} h \,\circ\,
 \geneta_{\bar{f}}\T_{k'}
\circ
 \mu_{k \circ \bar{f'}, \bar{g}}
 \,\circ\, \mu_{\bar{f'},k}\T_{\bar{g}} \,\circ\, \T_{\bar{f'}} \, \T_k \,
 \geneta_{\bar{g}} \,\circ\,
 \geneta_{\bar{f'}}\T_k \\
\plabel{$\geneta$ nat.}
%
%
\equiv \; &
\T_{\bar{g'} \circ k' \circ \bar{f}} h' \,\circ\,
 \mu_{k' \circ \bar{f}, \bar{g'}}
 \,\circ\, \mu_{\bar{f},k'}\T_{\bar{g'}} \,\circ\,
\T_{\bar{f}} \, \T_{k'} \,
\T_{\bar{g'}} h \,\circ\,
\T_{\bar{f}} \T_{\bar{g} \circ k \circ \bar{f'}}
 \geneta_{\bar{g'}}
\,\circ\,
 \geneta_{\bar{f}}\T_{k'}
\circ
 \mu_{k \circ \bar{f'}, \bar{g}}
 \,\circ\, \mu_{\bar{f'},k}\T_{\bar{g}} \,\circ\, \T_{\bar{f'}} \, \T_k \,
 \geneta_{\bar{g}} \,\circ\,
 \geneta_{\bar{f'}}\T_k \\
\plabel{$\mu$ nat.}
\equiv \; &
\T_{\bar{g'} \circ k' \circ \bar{f}} h' \,\circ\,
 \mu_{k' \circ \bar{f}, \bar{g'}}
\,\circ\,
\T_{k' \circ \bar{f}}
\T_{\bar{g'}} h
 \,\circ\, \mu_{\bar{f},k'}\T_{\bar{g'}} \,\circ\,
\T_{\bar{f}} \T_{\bar{g} \circ k \circ \bar{f'}}
 \geneta_{\bar{g'}}
\,\circ\,
 \geneta_{\bar{f}}\T_{k'}
\circ
 \mu_{k \circ \bar{f'}, \bar{g}}
 \,\circ\, \mu_{\bar{f'},k}\T_{\bar{g}} \,\circ\, \T_{\bar{f'}} \, \T_k \,
 \geneta_{\bar{g}} \,\circ\,
 \geneta_{\bar{f'}}\T_k \\
\plabel{$\mu$ nat.}
\equiv \; &
\T_{\bar{g'} \circ k' \circ \bar{f}} h' \,\circ\,
\T_{\bar{g'} \circ k' \circ \bar{f}} h
\circ
 \mu_{k' \circ \bar{f}, \bar{g'}}
 \,\circ\, \mu_{\bar{f},k'}\T_{\bar{g'}} \,\circ\,
\T_{\bar{f}} \T_{\bar{g} \circ k \circ \bar{f'}}
 \geneta_{\bar{g'}}
\,\circ\,
 \geneta_{\bar{f}}\T_{k'}
\circ
 \mu_{k \circ \bar{f'}, \bar{g}}
 \,\circ\, \mu_{\bar{f'},k}\T_{\bar{g}} \,\circ\, \T_{\bar{f'}} \, \T_k \,
 \geneta_{\bar{g}} \,\circ\,
 \geneta_{\bar{f'}}\T_k \\
\plabel{func.}
\equiv \; &
\T_{\bar{g'} \circ k' \circ \bar{f}} (h' \circ h)
\circ
 \mu_{k' \circ \bar{f}, \bar{g'}}
 \,\circ\, \mu_{\bar{f},k'}\T_{\bar{g'}} \,\circ\,
\T_{\bar{f}} \T_{\bar{g} \circ k \circ \bar{f'}}
 \geneta_{\bar{g'}}
\,\circ\,
 \geneta_{\bar{f}}\T_{k'}
\circ
 \mu_{k \circ \bar{f'}, \bar{g}}
 \,\circ\, \mu_{\bar{f'},k}\T_{\bar{g}} \,\circ\, \T_{\bar{f'}} \, \T_k \,
 \geneta_{\bar{g}} \,\circ\,
 \geneta_{\bar{f'}}\T_k \\
%
%
\plabel{$\mu$ assoc.}
\equiv \; &
\T_{\bar{g'} \circ k' \circ \bar{f}} (h' \circ h) \circ
 \mu_{\bar{f}, \bar{g'} \circ k'}
 \,\circ\, \T_{\bar{f}} \mu_{k',\bar{g'}} \,\circ\,
\T_{\bar{f}} \T_{\bar{g} \circ k \circ \bar{f'}}
 \geneta_{\bar{g'}}
\,\circ\,
 \geneta_{\bar{f}}\T_{k'}
\circ
 \mu_{k \circ \bar{f'}, \bar{g}}
 \,\circ\, \mu_{\bar{f'},k}\T_{\bar{g}} \,\circ\, \T_{\bar{f'}} \, \T_k \,
 \geneta_{\bar{g}} \,\circ\,
 \geneta_{\bar{f'}}\T_k \\
\plabel{$\mu$ assoc.}
%
\equiv \; &
\T_{\bar{g'} \circ k' \circ \bar{f}} (h' \,\circ\, h) \,\circ\,
 \mu_{\bar{f}, \bar{g'} \circ k'}
 \,\circ\,
\T_{\bar{f}} \, \mu_{k',\bar{g'}} \,\circ\,
\T_{\bar{f}} \T_{\bar{g} \circ k \circ \bar{f'}}
 \geneta_{\bar{g'}}
\,\circ\,
 \geneta_{\bar{f}}\T_{k'}
\,\circ\,
\mu_{\bar{f'}, \bar{g} \circ k}
\,\circ\,
\T_{\bar{f'}} \mu_{k,\bar{g}}
\,\circ\, \T_{\bar{f'}} \, \T_k \,
 \geneta_{\bar{g}} \,\circ\,
 \geneta_{\bar{f'}}\T_k \\
\plabel{$\geneta$ nat.}
%
\equiv \; &
\T_{\bar{g'} \circ k' \circ \bar{f}} (h' \,\circ\, h) \,\circ\,
 \mu_{\bar{f}, \bar{g'} \circ k'}
 \,\circ\,
\T_{\bar{f}} \, \mu_{k',\bar{g'}} \,\circ\,
\T_{\bar{f}} \T_{\bar{g} \circ k \circ \bar{f'}}
 \geneta_{\bar{g'}}
\,\circ\,
\T_{\bar{f}}
\mu_{\bar{f'}, \bar{g} \circ k}
\,\circ\,
\T_{\bar{f}} \T_{\bar{f'}} \mu_{k,\bar{g}}
\,\circ\,
 \geneta_{\bar{f}}\T_{f'}\T_{k}\T_{\bar{g}}
\,\circ\,
\T_{\bar{f'}} \, \T_k \,
 \geneta_{\bar{g}} \,\circ\,
 \geneta_{\bar{f'}}\T_k \\
\plabel{$\mu$ nat.}
%
\equiv \; &
\T_{\bar{g'} \circ k' \circ \bar{f}} (h' \,\circ\, h) \,\circ\,
 \mu_{\bar{f}, \bar{g'} \circ k'}
 \,\circ\,
\T_{\bar{f}} \, \mu_{k',\bar{g'}} \,\circ\,
\T_{\bar{f}}
\mu_{\bar{f'}, \bar{g} \circ k}  \T_{\bar{g'}}
\,\circ\,
\T_{\bar{f}} \T_{\bar{f'}} \T_{\bar{g} \circ k}
 \geneta_{\bar{g'}}
\,\circ\,
\T_{\bar{f}} \T_{\bar{f'}} \mu_{k,\bar{g}}
\,\circ\,
 \geneta_{\bar{f}}\T_{f'}\T_{k}\T_{\bar{g}}
\,\circ\,
\T_{\bar{f'}} \, \T_k \,
 \geneta_{\bar{g}} \,\circ\,
 \geneta_{\bar{f'}}\T_k \\
\plabel{$\mu$ nat.}
%
\equiv \; &
\T_{\bar{g'} \circ k' \circ \bar{f}} (h' \,\circ\, h) \,\circ\,
 \mu_{\bar{f}, \bar{g'} \circ k'}
 \,\circ\,
\T_{\bar{f}} \, \mu_{k',\bar{g'}} \,\circ\,
\T_{\bar{f}}
\mu_{\bar{f'}, \bar{g} \circ k} \T_{\bar{g'}}
\,\circ\,
\T_{\bar{f}} \T_{\bar{f'}} \mu_{k,\bar{g}} \T_{\bar{g'}}
\,\circ\,
\T_{\bar{f}} \T_{\bar{f'}} \T_{k} \T_{\bar{g}} \geneta_{\bar{g'}}
\,\circ\,
 \geneta_{\bar{f}}\T_{f'}\T_{k}\T_{\bar{g}}
\,\circ\,
\T_{\bar{f'}} \, \T_k \,
 \geneta_{\bar{g}} \,\circ\,
 \geneta_{\bar{f'}}\T_k
\\
\plabel{$k'$ def.}
\equiv \; &
\T_{\bar{g'} \circ \bar{g} \circ k \circ \bar{f'} \circ \bar{f}}
            (h'\circ h) \,\circ\,
\mu_{\bar{f}, \bar{g'} \circ \bar{g} \circ k \circ \bar{f'}}
\circ
\T_{\bar{f}} \mu_{\bar{g} \circ k \circ \bar{f'}, \bar{g'}}
\circ
\T_{\bar{f}} \mu_{\bar{f'},\bar{g} \circ k} \T_{\bar{g'}}
\circ
\T_{\bar{f}} \T_{\bar{f'}} \mu_{k, \bar{g}} \T_{\bar{g'}}
\circ
\T_{\bar{f}} \T_{\bar{f'}} \, \T_k \, \T_{\bar{g}} \geneta_{\bar{g'}} \circ
\geneta_{\bar{f}} \T_{\bar{f'}} \T_k \T_{\bar{g}} \circ
\T_{\bar{f'}} \, \T_k \, \geneta_{\bar{g}} \,\circ\,
\geneta_{\bar{f'}} \T_k
\\
\plabel{$\mu$ assoc.}
\equiv \; &
\T_{\bar{g'} \circ \bar{g} \circ k \circ \bar{f'} \circ \bar{f}}
            (h'\circ h) \,\circ\,
\mu_{\bar{f}, \bar{g'} \circ \bar{g} \circ k \circ \bar{f'}}
\circ
\T_{\bar{f}} \mu_{\bar{f'}, \bar{g'} \circ \bar{g} \circ k}
\circ
\T_{\bar{f}} \T_{\bar{f'}} \mu_{\bar{g} \circ k, \bar{g'}}
\circ
\T_{\bar{f}} \T_{\bar{f'}} \mu_{k, \bar{g}} \T_{\bar{g'}}
\circ
\T_{\bar{f}} \T_{\bar{f'}} \, \T_k \, \T_{\bar{g}} \geneta_{\bar{g'}} \circ
\geneta_{\bar{f}} \T_{\bar{f'}} \T_k \T_{\bar{g}} \circ
\T_{\bar{f'}} \, \T_k \, \geneta_{\bar{g}} \,\circ\,
\geneta_{\bar{f'}} \T_k
\\
\plabel{$\mu$ assoc.}
\equiv \; &
\T_{\bar{g'} \circ \bar{g} \circ k \circ \bar{f'} \circ \bar{f}}
            (h'\circ h) \,\circ\,
\mu_{\bar{f}, \bar{g'} \circ \bar{g} \circ k \circ \bar{f'}}
\circ
\T_{\bar{f}} \mu_{\bar{f'}, \bar{g'} \circ \bar{g} \circ k}
\circ
\T_{\bar{f}} \T_{\bar{f'}} \mu_{k, \bar{g'} \circ \bar{g} }
            \,\circ\,
\T_{\bar{f}} \T_{\bar{f'}} \, \T_k \, \mu_{\bar{g}, \bar{g'}} \circ
\T_{\bar{f}} \T_{\bar{f'}} \, \T_k \, \T_{\bar{g}} \geneta_{\bar{g'}} \circ
\geneta_{\bar{f}} \T_{\bar{f'}} \T_k \T_{\bar{g}} \circ
\T_{\bar{f'}} \, \T_k \, \geneta_{\bar{g}} \,\circ\,
\geneta_{\bar{f'}} \T_k
\\
\plabel{$\mu$ assoc.}
\equiv \; &
\T_{\bar{g'} \circ \bar{g} \circ k \circ \bar{f'} \circ \bar{f}}
            (h'\circ h) \,\circ\,
\mu_{\bar{f}, \bar{g'} \circ \bar{g} \circ k \circ \bar{f'}}
\circ
\T_{\bar{f}} \mu_{\bar{f'}, \bar{g'} \circ \bar{g} \circ k}
\circ
\T_{\bar{f}} \T_{\bar{f'}} \mu_{k, \bar{g'} \circ \bar{g} }
            \,\circ\,
\T_{\bar{f}} \T_{\bar{f'}} \, \T_k \, \mu_{\bar{g}, \bar{g'}} \circ
\geneta_{\bar{f}} \T_{\bar{f'}} \T_k \T_{\bar{g'} \circ \bar{g}}
\circ
\T_{\bar{f'}} \, \T_k \, \T_{\bar{g}} \geneta_{\bar{g'}} \circ
\T_{\bar{f'}} \, \T_k \, \geneta_{\bar{g}} \,\circ\,
\geneta_{\bar{f'}} \T_k
\\
\plabel{$\mu$ assoc.}
\equiv \; &
\T_{\bar{g'} \circ \bar{g} \circ k \circ \bar{f'} \circ \bar{f}}
            (h'\circ h) \,\circ\,
\mu_{\bar{f'} \circ \bar{f}, \bar{g'} \circ \bar{g} \circ k}
\circ
 \mu_{\bar{f}, \bar{f'}} \T_{\bar{g'} \circ \bar{g} \circ k} \circ
\T_{\bar{f}} \T_{\bar{f'}} \mu_{k, \bar{g'} \circ \bar{g} }
            \,\circ\,
\T_{\bar{f}} \T_{\bar{f'}} \, \T_k \, \mu_{\bar{g}, \bar{g'}} \circ
\geneta_{\bar{f}} \T_{\bar{f'}} \T_k \T_{\bar{g'} \circ \bar{g}}
\circ
\T_{\bar{f'}} \, \T_k \, \T_{\bar{g}} \geneta_{\bar{g'}} \circ
\T_{\bar{f'}} \, \T_k \, \geneta_{\bar{g}} \,\circ\,
\geneta_{\bar{f'}} \T_k \\
\plabel{$\mu$ nat.}
\equiv \; &
\T_{\bar{g'} \circ \bar{g} \circ k \circ \bar{f'} \circ \bar{f}}
            (h'\circ h) \,\circ\,
\mu_{\bar{f'} \circ \bar{f}, \bar{g'} \circ \bar{g} \circ k}
\circ
\T_{\bar{f'} \circ \bar{f}} \mu_{k, \bar{g'} \circ \bar{g} }
            \,\circ\,
 \mu_{\bar{f}, \bar{f'}} \T_k \T_{\bar{g'} \circ \bar{g}} \circ
\T_{\bar{f}} \T_{\bar{f'}} \, \T_k \, \mu_{\bar{g}, \bar{g'}} \circ
\geneta_{\bar{f}} \T_{\bar{f'}} \T_k \T_{\bar{g'} \circ \bar{g}}
\circ
\T_{\bar{f'}} \, \T_k \, \T_{\bar{g}} \geneta_{\bar{g'}} \circ
\T_{\bar{f'}} \, \T_k \, \geneta_{\bar{g}} \,\circ\,
\geneta_{\bar{f'}} \T_k
\\
\plabel{$\mu$ assoc.}
\equiv \; &
\T_{\bar{g'} \circ \bar{g} \circ k \circ \bar{f'} \circ \bar{f}} (h'\circ h) \,\circ\,
 \mu_{k \circ \bar{f'} \circ \bar{f}, \bar{g'} \circ \bar{g}}
 \,\circ\, \mu_{\bar{f'} \circ \bar{f},k}\T_{\bar{g'} \circ \bar{g}}
            \,\circ\,
 \mu_{\bar{f}, \bar{f'}} \T_k \T_{\bar{g'} \circ \bar{g}} \circ
\T_{\bar{f}} \T_{\bar{f'}} \, \T_k \, \mu_{\bar{g}, \bar{g'}} \circ
\geneta_{\bar{f}} \T_{\bar{f'}} \T_k \T_{\bar{g'} \circ \bar{g}}
\circ
\T_{\bar{f'}} \, \T_k \, \T_{\bar{g}} \geneta_{\bar{g'}} \circ
\T_{\bar{f'}} \, \T_k \, \geneta_{\bar{g}} \,\circ\,
\geneta_{\bar{f'}} \T_k
\\
\plabel{$\geneta$ nat.}
\equiv \; &
\T_{\bar{g'} \circ \bar{g} \circ k \circ \bar{f'} \circ \bar{f}} (h'\circ h) \,\circ\,
 \mu_{k \circ \bar{f'} \circ \bar{f}, \bar{g'} \circ \bar{g}}
 \,\circ\, \mu_{\bar{f'} \circ \bar{f},k}\T_{\bar{g'} \circ \bar{g}}
            \,\circ\,
 \mu_{\bar{f}, \bar{f'}} \T_k \T_{\bar{g'} \circ \bar{g}} \circ
\geneta_{\bar{f}} \T_{\bar{f'}} \T_k \T_{\bar{g'} \circ \bar{g}}
\circ
\T_{\bar{f'}} \, \T_k \, \mu_{\bar{g}, \bar{g'}} \circ
\T_{\bar{f'}} \, \T_k \, \T_{\bar{g}} \geneta_{\bar{g'}} \circ
\T_{\bar{f'}} \, \T_k \, \geneta_{\bar{g}} \,\circ\,
\geneta_{\bar{f'}} \T_k
\\
\plabel{func.}
\equiv \; &
\T_{\bar{g'} \circ \bar{g} \circ k \circ \bar{f'} \circ \bar{f}} (h'\circ h) \,\circ\,
 \mu_{k \circ \bar{f'} \circ \bar{f}, \bar{g'} \circ \bar{g}}
 \,\circ\, \mu_{\bar{f'} \circ \bar{f},k}\T_{\bar{g'} \circ \bar{g}}
            \,\circ\,
 \mu_{\bar{f}, \bar{f'}} \T_k \T_{\bar{g'} \circ \bar{g}} \circ
\geneta_{\bar{f}} \T_{\bar{f'}} \T_k \T_{\bar{g'} \circ \bar{g}}
\circ
\T_{\bar{f'}} \, \T_k \,
 (\mu_{\bar{g}, \bar{g'}} \circ \T_{\bar{g}} \geneta_{\bar{g'}} \circ \geneta_{\bar{g}}) \,\circ\,
\geneta_{\bar{f'}} \T_k
\\
\plabel{$\geneta$ nat.}
\equiv \; &
\T_{\bar{g'} \circ \bar{g} \circ k \circ \bar{f'} \circ \bar{f}} (h'\circ h) \,\circ\,
 \mu_{k \circ \bar{f'} \circ \bar{f}, \bar{g'} \circ \bar{g}}
 \,\circ\, \mu_{\bar{f'} \circ \bar{f},k}\T_{\bar{g'} \circ \bar{g}}
            \,\circ\,
 \mu_{\bar{f}, \bar{f'}} \T_k \T_{\bar{g'} \circ \bar{g}} \circ
 \T_{\bar{f}} \T_{\bar{f'}} \, \T_k \,
 (\mu_{\bar{g}, \bar{g'}} \circ \T_{\bar{g}} \geneta_{\bar{g'}} \circ \geneta_{\bar{g}}) \,\circ\,
\geneta_{\bar{f}} \T_{\bar{f'}} \T_k
\circ \geneta_{\bar{f'}} \T_k
\\
\plabel{$\mu$ nat.}
\equiv \; &
\T_{\bar{g'} \circ \bar{g} \circ k \circ \bar{f'} \circ \bar{f}} (h'\circ h) \,\circ\,
 \mu_{k \circ \bar{f'} \circ \bar{f}, \bar{g'} \circ \bar{g}}
 \,\circ\, \mu_{\bar{f'} \circ \bar{f},k}\T_{\bar{g'} \circ \bar{g}} \,\circ\, \T_{\bar{f'} \circ \bar{f}} \, \T_k \,
 (\mu_{\bar{g}, \bar{g'}} \circ \T_{\bar{g}} \geneta_{\bar{g'}} \circ \geneta_{\bar{g}}) \,\circ\,
 \mu_{\bar{f}, \bar{f'}} \T_k \circ \geneta_{\bar{f}} \T_{\bar{f'}} \T_k
\circ \geneta_{\bar{f'}} \T_k
\\
\plabel{$\geneta$ nat.}
\equiv \; &
\T_{\bar{g'} \circ \bar{g} \circ k \circ \bar{f'} \circ \bar{f}} (h'\circ h) \,\circ\,
 \mu_{k \circ \bar{f'} \circ \bar{f}, \bar{g'} \circ \bar{g}}
 \,\circ\, \mu_{\bar{f'} \circ \bar{f},k}\T_{\bar{g'} \circ \bar{g}} \,\circ\, \T_{\bar{f'} \circ \bar{f}} \, \T_k \,
 (\mu_{\bar{g}, \bar{g'}} \circ \T_{\bar{g}} \geneta_{\bar{g'}} \circ \geneta_{\bar{g}}) \,\circ\,
 \mu_{\bar{f}, \bar{f'}} \T_k \circ \T_f \geneta_{\bar{f'}} \T_k \circ
            \geneta_{\bar{f}} \T_k \\
\plabel{$\geneta$/$mu$}
\equiv \; &
\T_{\bar{g'} \circ \bar{g} \circ k \circ \bar{f'} \circ \bar{f}} (h'\circ h) \,\circ\,
 \mu_{k \circ \bar{f'} \circ \bar{f}, \bar{g'} \circ \bar{g}}
 \,\circ\, \mu_{\bar{f'} \circ \bar{f},k}\T_{\bar{g'} \circ \bar{g}} \,\circ\, \T_{\bar{f'} \circ \bar{f}} \, \T_k \,
 \geneta_{\bar{g'} \circ \bar{g}} \,\circ\,
 \mu_{\bar{f}, \bar{f'}} \T_k \circ \T_f \geneta_{\bar{f'}} \T_k \circ
            \geneta_{\bar{f}} \T_k \\
\plabel{$\geneta$/$mu$}
\equiv \; &
\T_{\bar{g'} \circ \bar{g} \circ k \circ \bar{f'} \circ \bar{f}} (h'\circ h) \,\circ\,
 \mu_{k \circ \bar{f'} \circ \bar{f}, \bar{g'} \circ \bar{g}}
 \,\circ\, \mu_{\bar{f'} \circ \bar{f},k}\T_{\bar{g'} \circ \bar{g}} \,\circ\, \T_{\bar{f'} \circ \bar{f}} \, \T_k \,
 \geneta_{\bar{g'} \circ \bar{g}} \,\circ\,
 \geneta_{\bar{f'} \circ \bar{f}}\T_k \\
\equiv \; & \Pm (f' \circ f, g' \circ g) (h' \circ h) \\
\end{align*}
\end{minipage}
}} \\

\item  $\Pm (\id_I, \id_I) \id_A = \id_{\Pm(I, I) A}$ follows
 from the right lax natural transformation axiom on $\eta$
 and the unit properties of cat-graded monads,
via the following reasoning (where $k = \inj_2 (I, I)$):
\begin{align*}
& \Pm (\id_I, \id_I) \id_A \\
\equiv \; & \T_{\bar{\id} \circ k \circ \bar{\id}} id \,\circ\,
\mu_{k \circ \bar{\id}, \bar{\id}}
\,\circ\, \mu_{\bar{\id},k}\T_{\bar{\id}} \,\circ\, \T_{\bar{\id}} \, \T_k \,
\geneta_{\bar{\id}} \,\circ\,
\geneta_{\bar{\id}}\T_k \\
\plabel{$\geneta/\id$ property}
\equiv \; & \T_{k} id \,\circ\,
\mu_{k, \bar{\id}}
\,\circ\, \mu_{\bar{\id},k}\T_{\bar{\id}} \,\circ\, \T_{\bar{\id}} \, \T_k \,
\eta_{I} \,\circ\,
\eta_{I}\T_k\; \\
\plabel{$\mu$ naturality}
\equiv \; & \T_{k} id \,\circ\,
\mu_{k, \bar{\id}}
\,\circ\, \eta_{I} \circ \mu_{\bar{\id},k}\T_{\bar{\id}} \,\circ\,
\eta_{I}\T_k \\
\plabel{$\mu$ unitality}
\equiv \; & \T_{k} id \,\circ\,
\mu_{k, \bar{\id}}
 \,\circ\,
\eta_{I}\T_k \\
\plabel{$\mu$ unitality}
\equiv \; & \T_{k} id \\
\plabel{functor identity property}
\equiv \; & id_{\T_k} \\
\equiv \; & id_{\Pm (I, I) A}
\end{align*}
\end{itemize}

\noindent
 The parameterised monad operations follow from
 the cat-graded monad $\mu^{\Pm}_{I,J,K} =
 \mu_{\inj_2 (I,J), \inj_2 (J,K)}$ and generalised unit $\eta^{\Pm}_{I} =
 \geneta_{\inj_1 \id_I}$.  This mapping is inverse to the former, e.g., with
 $\eta^{\Pm}_I = \geneta_{\inj_1(\id_I)} = \Pm id_I id_I
 \circ \eta^{\Pm}_I = \eta^{\Pm}_I$ ($\mu$ is more direct).
 For $\Pm \, (f, g) \, h$, substituting
 eq.~\ref{eq:geneta-functor-morphism} into
 eq.~\ref{eq:functor-converse-geneta} and applying
 the category-graded monad laws yields $\Pm (f, g) h$ (calculation
 elided for brevity). Therefore,
 there is just one category-graded monad with
 generalised unit for every parameterised monad.
\end{proof}

\end{document}